\documentclass[ejs]{imsart}

\RequirePackage[OT1]{fontenc}
\RequirePackage{amsthm,amsmath}
\RequirePackage[numbers]{natbib}
\RequirePackage[colorlinks,citecolor=blue,urlcolor=blue]{hyperref}
\usepackage{graphicx}

\startlocaldefs
\numberwithin{equation}{section}
\theoremstyle{plain}
\newtheorem{theorem}{Theorem}[section]
\newtheorem{definition}{Definition}[section]

\newtheorem{proposition}{Proposition}[section]
\newtheorem{lemma}{Lemma}[section]

\newtheorem{example}{Example}[section]
\newtheorem{remark}{Remark}[section]

\newtheorem{assumption}{Assumption}[section]

\newtheorem{condition}{Condition}[section]
\newtheorem{step}{Step}[section]
\endlocaldefs

\begin{document}

\begin{frontmatter}
	\title{Estimation for recurrent events through conditional estimating equations}
	\runtitle{Estimation for recurrent events through estimating equations}
	
	\begin{aug}
		\author{\fnms{Ioana} \snm{Schiopu-Kratina}\thanksref{t1}\ead[label=e1]{ioanakratina@gmail.com}}
		\and
		\author{\fnms{Liu} \snm{Hai Yan}\thanksref{t1}\ead[label=e2]{hliu038@uottawa.ca}}
		
		\address{\printead{e1,e2}}
		
		\author{\fnms{Mayer} \snm{Alvo} \thanksref{t1}\ead[label=e3]{malvo@uottawa.ca}}
		\and
		\author{\fnms{Pierre-J\'er\^ome} \snm{Bergeron} \thanksref{t1}\ead[label=e6]{pierrejerome.bergeron@gmail.com}}
			\address{\printead{e3,e6}}
		
		\thankstext{t1}{Department of Mathematics and Statistics, University of Ottawa, 150 Louis Pasteur Private, Ottawa ON K1N 6N5}
		\runauthor{Schiopu-Kratina et al.}
		
		\affiliation{University of Ottawa}
		
	\end{aug}

\begin{abstract}
We present new estimators for the statistical analysis of the dependence of the mean gap time length between consecutive recurrent events, on a set of explanatory random variables and in the presence of right censoring. The dependence is expressed through regression-like and overdispersion parameters, estimated via conditional estimating equations. The mean and variance of the length of each gap time, conditioned on the observed history of prior events and other covariates, are known functions of parameters and covariates. Under certain conditions on censoring, we construct normalized estimating functions that are asymptotically unbiased and contain only observed data. We discuss the existence, consistency and asymptotic normality of a sequence of estimators of the parameters, which are roots of these estimating equations. Simulations suggest that our estimators could be used successfully with a relatively small sample size in a study of short duration.

\end{abstract}

\tableofcontents
\end{frontmatter}

\section{Introduction}\label{introduction}
\subsection{Background}

Recurrent event data for a subject consist of repeated occurrences of the same type of event over  a period of time. Some examples from health sciences given in \cite{Amorim and Jianwen} and \cite{Rogers.etal:2016} are: heart failure hospitalizations in cardio-vascular trials, seizures in epilepsy trials, asthma attacks, migraines, cancer recurrences \textit{etc.}  In many medical studies, the focus is on estimating counts of recurrent events, or the time to the event, as in survival analysis.  Some statistics of interest related to recurrent event data are: the total number of events over a specified time period, the rate at which events occur, time to the event for successive events, and gap times between two successive events. A problem of great practical importance in epidemiology is the estimation of the time (or its average) between two successive mutations of a specific virus. Studies based on gap times  are also of interest in the study of system failures or cyclical phenomena, like hurricanes or earthquakes in specified geographic areas, where it is of interest to characterize the cycle length.

The authors of \cite{clement and Strawderman} have identified two classes of methods of analysis for recurrent event data: intensity and marginal models. They refer to \cite{rJ07} for an extensive review of the existing methods. They also give references on the use of marginal models. We discuss marginal models in connection with our work further on in this section.

Statistical methods for time to event data are well established. The time to event studies often rely on Cox proportional-hazards model, the most commonly used regression model in survival analysis. The drawback of using this method with recurrent event data is that it applies to independent gap times, which generally limits its use to a one-time event. The Andersen-Gill generalization in \cite{Andson and Gill:1982} gives the intensity/hazard rate for recurrent event processes and is such that each gap time contributes to the likelihood. Furthermore, gap times are conditionally independent, given the covariates. This method can be used if the correlation between events is induced by covariates. It is a semiparametric method in that the baseline intensity is not required. The method employs, however, a likelihood function.

Recurrent event data are collected on a preferably large sample of subjects. We view it as longitudinal data collected on each subject before or at the time each recurrent event occurs. Data collected on the same subject on different occasions are often correlated and therefore can be viewed as a cluster of observations when an intra-cluster dependence exists. Right censoring creates, for many subjects, partially observed last gap times. Even if each subject were observed for a prespecified number of complete gap times, the drop-out of subjects would create partially observed last gap times. Deleting these gap times leads to length bias in the analysis and generates incomplete data sets. The commonly accepted practice is to replace the incompletely observed gap times with data that does not alter the distribution of the fully observed gap times, \textit{i.e.}, to impute.

A typical example is given in \cite{Murphy.etal:1995}. The authors study the menstrual pattern of a sample of Lese women from Zaire.  They express the mean length of the menstrual cycle, conditioned on the past cycles length and covariates (age, Body Mass Index, or BMI), as a function of these variables and a vector valued parameter. Likewise, the conditional variance is parameterized with overdispersion as an additional component of the parameter. This study encapsulates most of the challenges encountered in the analysis of recurrent event data: the covariates are time varying and there is right censoring.

In this article, we base our approach on the methodology developed in \cite{Murphy.etal:1995} and \cite{clement and Strawderman},  which does not specify the joint probability of repeated occurrences, given the previous history, \textit{i.e.}, we use a semiparametric approach. The approach in \cite{Murphy.etal:1995} and \cite{clement and Strawderman} for the analysis of recurrent events is based on the use of estimating functions, introduced in the seminal paper \cite{Liang and Zeger:1986} for marginal longitudinal models. The theoretical background developed in \cite{Murphy and Li:1995} accommodates the drop-out of subjects from the study, through the use of conditional estimating functions. As expressed in \cite{Murphy and Li:1995}, theoretically, this conditioning amounts to a projection of a partial likelihood score function onto the space generated by a class of ``conditionally linear'' estimating functions. As shown in \cite{Murphy and Li:1995}, one can remove certain terms from this projection, without compromising its inferential properties. Conditioning also helps mitigate problems encountered when dealing with time-varying covariates

We have identified three stages in the development of the methodology on which our work is based. Firstly, thanks to the broad applicability of the method described in \cite{Liang and Zeger:1986}, the estimating equations methodology could be applied to recurrent event processes in \cite{Murphy.etal:1995} and \cite{clement and Strawderman}. Secondly, the extension to the case of random covariates posed some problems, which have been circumvented. One such problem is the potential bias the random covariates may  introduce in the generalized estimating functions. The unbiasedness of these functions plays an important role in establishing the consistency of the estimators that they define. It is at this stage that conditional estimating functions were introduced in
\cite{Murphy and Li:1995}, applied in \cite{Murphy.etal:1995} and then in \cite{clement and Strawderman}. These are unbiased estimating functions, conditioned on sigma fields that capture the evolution in time of the recurrent event process. Additional conditions  are required to arrive at a final conditioned function used in estimation. Such a condition is given in \cite{clement and Strawderman}, but not in \cite{Murphy.etal:1995}. Neither reference provides an explicit definition of the $\sigma$-field used for conditioning. We provide such a description in Section 6.2 and elaborate on our choice in Section 6.1.

The last stage consists of selecting an imputation method and defining estimators (implicitly) from estimating equations containing the imputed data.
 A parametric method is used in \cite{Murphy.etal:1995} and \cite{clement and Strawderman}. We use  only the observed data to define our estimators of the regression and overdispersion parameters. Conditions for the existence, consistency and asymptotic normality of these estimators are given in \cite{clement and Strawderman}. We cover these properties in Section 3--4 of this article. In \cite{Murphy.etal:1995} the estimators are produced using an iterative process. No formal proofs of their properties are given. More on the comparison of methods is in our Section 6.2.

We now present some technical aspects related to the three stages described above.
We briefly describe the methodology in \cite{Liang and Zeger:1986}, with their notation.

Let $y_{ij}\in R$ represent the response variables from subjects, indexed by $i$, $x_{ij}\in R^p$ the \textit{nonrandom} covariates and $y_i$ the vector with components \\ $ \displaystyle y_{ij}\;(j=1,\ldots,m_i, i=1,\ldots$), where $j$ labels the occasion on which measurements are taken and $i$ the subject of the analysis.  The vectors $y_i$ are assumed independent, $(i=1,\ldots).$  Here $\beta\in R^{p}$ is the  regression parameter. Only the means and variances at each occasion are specified, namely
\begin{equation}
\label{equation0.1}
E_{\beta}[y_{ij}]=\mu(x_{ij}^T\beta)=\mu_{ij}(\beta),\quad Var_{\beta}[y_{ij}]=\dot{\mu}(x_{ij}^T\beta)=\phi \sigma_{ij}^2(\beta),
\end{equation}
where $\mu$ is a known canonical link function, $\dot{\mu}$ its first derivative and $\phi$ denotes the overdispersion parameter. We denote by $\mu_{i}(\beta)$ the vector with components $\mu_{ij}(\beta), (j=1,\ldots,m_i, i=1,\ldots)$.
The models defined in (\ref{equation0.1}) are called \textit{marginal models} because the first two moments of the response variables are specified separately for each occasion $(j=1,\ldots, m_i)$.

To estimate the true value of the regression parameter denoted by $\beta_0$, the authors of \cite{Liang and Zeger:1986} obtained a sequence of estimators $\hat{\beta}_n$, which are roots of the generalized estimating equation
\begin{equation}
\label{equation0.2}
\sum_{i=1}^n\bigg(\frac{\partial \mu_{i}(\beta)}{\partial \beta^T}\bigg)^TV_{i}^{-1}(\beta, \alpha)(y_i-\mu_{i}(\beta))=0 \quad (n=1,\ldots).
\end{equation}
In (\ref{equation0.2}), 
$\alpha$ is a nuisance parameter  and $V_{i}(\beta, \alpha)$ is a $m_i\times m_i$ \textit{working covariance } matrix, which stands for the correct, but unknown intra-cluster covariance matrix. Under the \textit{working independence assumption}, the correlation matrix corresponding to $V_i$ in (\ref{equation0.2}) is the identity matrix.
It is shown in \cite{Liang and Zeger:1986} that the sequence $\{\hat{\beta}_n\}_{n\geq 1}$ is strongly consistent, $i.e.$, it converges to $\beta_0$ $a.s.$, regardless of what the working covariance is, and this sequence of estimators is asymptotically normally distributed. The penalty for using a working covariance in lieu of the true one is a decrease in efficiency. Theoretical justification and extensions of the results in \cite{Liang and Zeger:1986} were given in \cite{Xie and Yang:2003}, \cite{Balan.etal:2005}, \cite{Balan.etal:2010}. Further applications and examples can be found in \cite{Zeger and Liang: 1986} and \cite{Zeger and Liang: 1992}.

We now present the use of the working independence assumption in the analysis of recurrent events (see also section 2.5 of \cite{clement and Strawderman}). We first point out the difference in notation. In \cite{clement and Strawderman}
and in this article, the response process is denoted $Y_{ij}$, and it represents a measure of the gap time between the $(j-1)$th and the $j$th occurrence of the event for subject $i$. The main regression parameter $\beta$ is denoted $\theta$, 
 whereas $\sigma_{ij}(\beta)$ in (\ref{equation0.1}) is $V_{ij}(\theta)$ in \cite{clement and Strawderman}, $i.e.$, the marginal standard deviation. Censoring aside, the first expression in (2.7) of \cite{clement and Strawderman} generates estimating equations which correspond to (\ref{equation0.2}) here, with $Z_{ij}(\theta)$ and $f_{ij}(\theta)$ defined in (2.6) of \cite{clement and Strawderman} and in  (\ref{equation1.7}) here.

We now discuss the extension of (\ref{equation0.1}) to the case of random covariates, which requires some form of conditioning in (\ref{equation0.1}). The simplest situation occurs when conditioning is done on the last observed covariate $x_{ij}$ (marginal models). As shown in \cite{peper and Anderson:1994} and \cite{Lai and Small:2007}, when some covariates are random and \textit{time-varying}, the expectation of the estimating functions in (\ref{equation0.2}) may not be zero, so the generated estimators may not be consistent.  This does not happen under the working independence assumption. It is suggested in \cite{peper and Anderson:1994} to either use the working independence assumption or a strong condition of independence among covariates, not needed in \cite{Yi.et al: 2012}.
 In this article we adopt the working independence method. The model assumptions (\ref{equation0.1}) have been relaxed here, in that the marginal, conditional variance need not be related to the conditional mean.
In the context of recurrent events, conditioning in (\ref{equation0.1}) is done on the data of each subject at the $jth$ event, available right after the $(j-1)th$ occurrence of the event $(j=1,\ldots )$.

Time varying covariates generate a stochastic process, which may interact with the recurrent event process in complex ways.
 In \cite{Murphy.etal:1995}, a hormonal level of a woman would be a time-varying covariate, which changes within cycles. A slow varying covariate in this study is age. In \cite{Murphy.etal:1995}, as in here, we consider covariates which are relatively constant within gap times, but vary between gap times. We also asssume that changes within a gap time do not affect its length.

Conditioning on some variables may control the within cycles variability of time-varying covariates. In \cite{Murphy.etal:1995}, the BMI is measured at the begining of a period, so its value depends on the length of the previous cycles. The conditioning $\sigma$-field for the mean in our  (\ref{equation1.4}) of Example \ref{example1.4} contains the previous cycles total length. Conditional estimating functions lead to a conditional imputation method. The use of observed data in the conditional $\sigma$- fields points to a missing at random (MAR, as in \cite{Little and Rubin: 2002}) nonresponse mechanism (p.454 in \cite{clement and Strawderman}). Other advantages  of conditioning in this context are described on p.1845 of \cite{Murphy.etal:1995}.

\subsection{Our contribution}

 Our main contribution consists of new estimators for the analysis of the dependence of the gap time between consecutive occurrences of a recurrent event on a set of random covariates, in the presence of right censoring. The estimators are defined as roots of estimating equations, in which pertinent data collected from all subjects replaces the partially censored gap times. By contrast, the imputation procedure in \cite{Murphy.etal:1995} and \cite{clement and Strawderman} relies on a parametric approach.

 To arrive at our estimating equations, we defined first two sets of estimating functions. The first set corresponds to the case of fully observed gap times, the second is a projection of the first set onto $\sigma$-fields generated by some observed data, which reflect the evolution in time of the censored process. The last terms of these conditioned functions contain moments of the unknown conditional distribution of the censored gap times. It is these last terms that are imputed, thus creating the final set of estimating functions from which the estimators are derived.

We attempted here to combine the pioneering ideas in \cite{Murphy and Li:1995} and \cite{Murphy.etal:1995} with the mathematical rigour (and notation) in \cite{clement and Strawderman}. This means, among other things, that we defined the conditioning $\sigma$-fields in \cite{Murphy.etal:1995} (identical to ours), as well as in \cite{clement and Strawderman} (see Section 6.2). We gave conditions for and proved the unbiasedness of all estimating functions that we employed. Our conditions model the censoring mechanism by expressing the degree of independence between some information on the censoring time and the recurrent event process. We also analyzed the asymptotic behaviours of our estimating functions and gave complete proofs of the asymptotic normality of our estimators (Appendix \ref{appb.2}).  Our Theorem \ref{theorem2} generalizes the results on the existence and consistency of the estimators presented in \cite{clement and Strawderman}. We used the survey in \cite{Murphy.etal:1995} to exemplify some  analytical conditions required for our results to hold. This study was also the methodological source of our numerical results. A comparison with the parametric imputation method in \cite{clement and Strawderman} shows that our nonparametric method performs well, specially when the sample of subjects is small and the length of the study is relatively short.

This article is organized as follows. Section 2 presents the model and our suggested estimating functions. Section 3 is about the strong consistency of our estimators, while their asymptotic normality is discussed in Section 4. In Section
5 we present simulation results. In Section 6 we compare methods which use conditional estimating functions and draw conclusions. Appendix \ref{app} contains a summary of some analytical results required in Section 3 and an illustrative example. The proofs required in Section 4 are in Appendix \ref{appb.2}. More simulation results  can be found  in the Supplementary section.

\section{The model and basic assumptions}\label{themodelassuption}

\subsection{Model assumptions and examples}
 We introduce the set-up and the estimating functions defined in \cite{clement and Strawderman}, with a slightly different notation. We then state the conditional model used and the conditional independence assumptions governing the censoring times.

 Data $d_i$ are collected on subject $i$ and  are generated by a distribution indexed by a true parameter $\eta_0=(\theta_0, \sigma_{0}^2)^T \in K \subset R^{p+1}$, where $\sigma_{0}^2\in R$ and $K$ is a compact set of parameters. Let $ (\Omega, \mathcal{F}) $ be a measurable space, $i.e.$, $\Omega$ is the set of all outcomes and $\mathcal{F}$ a $\sigma$-field. All our functions are $\mathcal{F}$- measurable and various probability measures $P_{\eta}, \eta \in K$ can be defined on this space. We assume that the time of origin for the analysis is $S_{i0}=0$, with subsequent similar events occurring at times
  $0<S_{i1}<\cdots<S_{ij}<\cdots.$
Observation stops at a time $0<C_i<\infty$ $(i=1,\ldots)$. The observed, uncensored data for subject $i$ at times $S_{ij}$ generate the $\sigma$-field
\begin{equation}
 \label{equation1.1}
  \mathcal{F}_{ij}=\sigma \{S_{il}\;(l=0,\ldots,j);x_{il}\; (l=1,\ldots,j+1)\}\quad (i=1,\ldots; j=0,\ldots),\nonumber
\end{equation}
where $x_{ij}$ denotes the covariate information, available at time $S_{i,j-1}$ $(j=1,\ldots)$. 
   Note that $\{\mathcal{F}_{ij}\}$ is a filtration, meaning that $\mathcal{F}_{ij}\subset\mathcal{F}_{i,j+1}$ $(j=0,\ldots)$.
We now define a filtration larger than $\{\mathcal{F}_{ij}\}\; (i,j=1,\ldots)$, which contains some information about the censoring times $C_i\; (i=1,\ldots)$.
Let
\begin{equation}
\mathcal{G}_{ij}=\sigma\{\mathcal{F}_{ij},\{S_{ik}\leq C_i\},\{S_{ik}=C_i\},(k=1,\dots,j)\}, \mathcal{G}_{i0}= \mathcal{F}_{i0}.\nonumber
\end{equation}
\begin{remark} As in \cite{Murphy.etal:1995}, p.1845, we assume that the time of departure of a subject from the study (``death") depends on outside influences, or on past covariates. If so, taking it as censoring time $C_i$ does not invalidate our analysis.
\end{remark}
Define the $j$th gap time
  $\displaystyle Y_{ij}=S_{ij}-S_{i,j-1} \;(i,j=1,\ldots).$
 We assume throughout that each $Y_{ij}$ has a finite second moment.

The data vectors are  assumed to be independent and identically distributed (\textit{i.i.d.}) and can be written as
\begin{equation}
 d_{i}^T=\{S_{ij}, x_{ij},C_i \;(j=1,\ldots)\},\nonumber
 \end{equation}
  where $x_{ij}$ are random covariates $ (i,j=1,\ldots)$.
With $\displaystyle \eta=(\theta^T, \sigma)^T\in K ,$ the basic modeling assumptions are
\begin{equation}
  \label{euation1.2}
  E_{\theta}[Y_{ij}\mid\mathcal{F}_{i,j-1}]=\mu_{ij}(\theta),\quad   \quad   var_{\eta}[Y_{ij}\mid \mathcal{F}_{i,j-1}]=\sigma^{2}V_{ij}^{2}(\theta),
\end{equation}
where $\mu_{ij}(\theta)\in R$ and $V_{ij}(\theta)>0 $ $(i,j=1,\ldots)$ are known scalar functions of the parameter vector $\theta\in R^p$, and of covariates, which will be displayed in examples.
 We think of $\theta$ as the regression parameter and of $\sigma^2$ as the overdispersion parameter.
While (\ref{euation1.2}) holds for all possible values of the parameter $\eta$, convergence results, which are used to estimate the true parameter $\eta_0$, hold in  the probability measure $P_{\eta_0}$ , which we omit writing, when no confusion may arise.

Unless specified otherwise we assume (A) throughout.
\begin{assumption}[$A$]
\begin{equation}
E_{\theta}[Y_{ij}\mid\mathcal{G}_{i,j-1}]=E_{\theta}[Y_{ij}\mid\mathcal{F}_{i,j-1}],\nonumber \end{equation}
\begin{equation}
var_{\eta}[Y_{ij}\mid\mathcal{G}_{i,j-1}]=var_{\eta}[Y_{ij}\mid\mathcal{F}_{i,j-1}]\; \;\quad(i,j=1,\ldots).\nonumber
\end{equation}
\end{assumption}

This assumption holds when, conditional on $\mathcal{F}_{i,j-1}$, the sigma fields generated by $\{S_{ik}\leq C_{i}\},\{S_{ik}=C_i\}\; (k=1,\ldots,j-1)$ and $Y_{ij}$ are independent $(i,j=1,\ldots)$ ( see 34.11 of \cite{Billingsley:1995}).

We now write condition $(A0)$ of \cite{clement and Strawderman} (p.465) in our notation. Let $H_{ij}$ be the $\sigma$-field generated by $\mathcal{F}_{i,j-1}$ and $\{C_i\geq S_{i,j-1}\}, \; i, j\geq 1.$ Then  $(A0)$ is:
\begin{eqnarray}
  E_{\theta}[Y_{i,j}|H_{i, j-1}]&=&E_{\theta}[Y_{ij}|\mathcal{F}_{i,j-1}]\nonumber\\
                      var_{\eta}[Y_{ij}|H_{i, j-1}]&=&var_{\eta}[Y_{ij}|\mathcal{F}_{i,j-1}]\nonumber
\end{eqnarray}

\begin{remark}
\label{remark2.1final}
Condition (A) is a stronger  variation of the noninformative censoring condition $(A0)$. While $(A0)$ holds when, in the presence of the covariate history up to and including occasion $j-1$ for subject $i$, $i.e.$, $\mathcal{F}_{i, j-1}$, the position of $ C_i$ with respect to $S_{i, j-1}$ has no bearing on the value of $Z_{ij}$. Our condition  $(A)$ holds if, in addition, the position of $C_i$ versus earlier occasions $S_{i, k},\; k\leq j-1$ is also noninformative for $Z_{ij}$. This allows us to define the filtration $\{\mathcal{G}_{ij}\}_{j\geq 1}$, on which we can obtain stronger asymptotic results.
\end{remark}
Our theoretical results apply to the following general examples.
   \begin{example}\label{example1.1} $E_{\theta}[Y_{ij}\mid\mathcal{F}_{i,j-1}]=\mu_{ij}(\theta), \quad var_{\eta}[Y_{ij}\mid\mathcal{F}_{i,j-1}]=\sigma^{2}V_{ij}^{2}(\theta)$.
   \end{example}
   \begin{example}\label{example1.2} $ E_{\theta}[Y_{ij}\mid\mathcal{F}_{i,j-1}]=\mu_{ij}(\theta),\quad var_{\eta}[Y_{ij}\mid\mathcal{F}_{i,j-1}]=\sigma^{2}\mu_{ij}^{2}(\theta)$.
   \end{example}

When no confusion may arise, we omit writing the subscript of $E$.

  Example \ref{example1.2} generalizes the accelerated gap times model proposed in \cite{Strawderman:2005}, which assumes that the gap times of the recurrent event process satisfy  $S_{ij}-S_{i,j-1}=R_{ij}\mu(\theta)$, where $R_{ij} $ are \textit{i.i.d.} random variables. Here $\mu(\theta)$, which is known, accelerates or decelerates the baseline gap times, based on the values of the time-independent covariates. When $E[R_{ij}]=1$, Example \ref{example1.2} is a direct generalization of the accelerated gap times model with $\mu_{ij}(\theta)= \mu(\theta), V_{ij}(\theta)=\mu(\theta)$ and $\sigma^{2}=var[R_{ij}]\;(i,j=1,\ldots)$.

We define, for $\theta^T=(\gamma_0, \gamma_1, \rho)\in R^3$ and $\rho\neq -(j-2)^{-1}$, the function
\begin{equation}
\label{equation1.3}
f_{j}(\rho)=\rho[\rho(j-1)+1-\rho]^{-1}\quad (j=1,\ldots).
\end{equation}

The example below was first introduced in \cite{Murphy.etal:1995}.
\begin{example}\label{example1.4} We assume that $\mu_{ij}(\theta)$ and $V_{ij}(\theta)$ are, for each $i=1,\ldots,$
\begin{equation}
\mu_{i1}(\theta)=\gamma_0+\gamma_1\overline{BMI}_{i1}\nonumber
\end{equation}
\begin{equation}
\label{equation1.4}
\mu_{ij}(\theta)=\gamma_0+\gamma_1\overline{BMI}_{ij}+f_{j}(\rho)
\bigg[\sum_{l=1}^{j-1}{Y_{il}}-\sum_{l=1}^{j-1}(\gamma_0+\gamma_1\overline{BMI}_{il})\bigg]\quad (j=2,\ldots),
\end{equation}
\begin{equation}
\label{equation1.5}
V_{ij}(\theta)= \bigg|1+f_{j}(\rho)\bigg|^{1/2}.
\end{equation}
In lieu of (\ref{equation1.5}), we could use $V_{ij}(\theta)=\mid\mu_{ij}(\theta)\mid$.
 The covariate $\overline{BMI}_{il}$ represents an average of several body mass index measurements taken on individual $i$ at the $(l-1)$th occurrence of the event. Formula (\ref{equation1.4}) is (3.3) of \cite{clement and Strawderman}, with the constant 28 incorporated in our $\gamma_0$. Let
   \begin{equation}
   c_{i1,0}(\theta)=\gamma_0;\;\;\; c_{i1,1}=\gamma_1; c_{i1,h}(\theta)=0 \;\;(h=2,\ldots).\nonumber
  \end{equation}
    For $x_{i1}$, the components are
 \begin{equation}
  x_{i1,0}=1;\;\; x_{i1,1}=\overline{BMI}_{i1}; \;\; x_{i1,h}=0 \; \; (h=2,\ldots).\nonumber
  \end{equation}
   For each $j=2,\ldots,$ let $c_{ij}(\theta)$ be the vector with components:
  \begin{eqnarray}
  c_{ij,0}(\theta)&=&\gamma_0-(j-1)\gamma_0f_{j}(\rho); \quad   c_{ij,h}(\theta)=-\gamma_1f_{j}(\rho)\;(h=1,\ldots, j-1); \nonumber\\
  c_{ij,j}(\theta)&=&\gamma_1; \quad c_{ij,h}(\theta)=f_{j}(\rho)\quad  (h=j+1,\ldots, 2j-1)\quad  \mbox{and}\nonumber\\
   c_{ij,h}(\theta)&=&0\quad (h=2j,\ldots).\nonumber
  \end{eqnarray}
   We define the components of the corresponding vector $ x_{ij}$ of covariates as :
    \begin{eqnarray}
    &&x_{ij,0}=1; \quad  x_{ij,h}=\overline{BMI}_{i,h}  (h=1,\ldots,j); \nonumber\\
    && x_{ij,h}=Y_{i,h-j} (h=j+1,\ldots,2j-1)\;\mbox{ and} \;x_{ij,h}=0 \; (h=2j,\ldots).\nonumber
    \end{eqnarray}
     One can see that $\displaystyle \mu_{ij}(\theta)=c_{ij}^T(\theta)x_{ij} \;(i,j=1,\ldots).$
\end{example}
The next example serves  as the basis  for the theoretical results illustrated in this article.
\begin{example} \label{example1.5}The conditional mean in (\ref{euation1.2}) is
$
\displaystyle \mu_{ij}(\theta)=\mu(c_{ij}^T(\theta)x_{ij})$\\$\displaystyle(i,j=1,\ldots),$
where $\theta$ is a p-dimensional vector, $\mu: R\rightarrow R$ is known, $ c_{ij}(\theta)$ and $x_{ij}$ are vectors of the same dimension.  The third derivatives of the function $\mu$ and of the components of $c_{ij}(\theta)$ are continuous. Here $c_{ij}(\theta)$ is a known function, which need not be a linear function of $\theta$.
\end{example}

\subsection{Unbiased estimating functions}

In this section we introduce the observed and the empirical estimating functions, which will be used throughout, and we study their unbiasedness properties.
We often appeal to the strong law of large numbers (SLLN) for $i.i.d.$ random variables and we always assume that these variables have finite expectations. As in \cite{clement and Strawderman}, we adopt the following notation:
\begin{equation}
\label{equation1.7}
  f_{ij}(\theta)=\frac{\partial\mu_{ij}(\theta)}{\partial\theta}V_{ij}^{-1}(\theta),\quad\quad  Z_{ij}(\theta)=\frac{Y_{ij}-\mu_{ij}(\theta)}{V_{ij}(\theta)}\quad (i,j=1,\ldots).
\end{equation}

We assume that $E[Z_{ij}^2(\theta)]<\infty\ (i,j=1,\ldots)$.  From (\ref{euation1.2}) and (\ref{equation1.7}), it follows that $f_{ij}(\theta)$ is
$ \mathcal{F}_{i,j-1}$-measurable. We consider the estimating functions
 \begin{equation}
\label{equation1.9}
g_{n,1}(\theta)=\sum_{i=1}^n\sum_{j=1}^{\infty}f_{ij}(\theta)Z_{ij}(\theta)I\{S_{i,j-1}<C_{i}\},
\end{equation}
\begin{equation}
\label{equation1.10}
g_{n,2}(\eta)=\sum_{i=1}^n\sum_{j=1}^{\infty}b_{ij}(\eta)(Z_{ij}^2(\theta)-\sigma^2)I\{S_{i,j-1}< C_i\}\quad (n=1,\ldots),
\end{equation}

where we assume that the inner sums are finite. In (\ref{equation1.10}) $b_{ij}(\eta)$ are $\mathcal{F}_{i,j-1}$- measurable random variables $(i,j=1,\ldots).$
\begin{definition}
\label{definitionS2.1}
Let $\tau_{i}=\min\{j\geq 1:C_i\leq S_{ij}\}$, if such $j$ exists, and $\tau_{i}=\infty$ otherwise $(i=1,\ldots)$.
 \end{definition}

Since $\{\tau_i=n\}\in \mathcal{G}_{i,n}$, $\tau_{i}$  is a stopping time with respect to the filtration $\{\mathcal{G}_{ij}\}_{j\geq 1}\;$ for $ (i=1,\ldots)$.
Now (\ref{equation1.9}--\ref{equation1.10}) become
 \begin{equation}
 \label{sequation.2.8newf}
g_{n,1}(\theta)=\sum_{i=1}^n\sum_{j=1}^{\tau_{i}}f_{ij}(\theta)Z_{ij}(\theta),
\end{equation}
\begin{equation}
\label{sequtaion.2.9.newf}
g_{n,2}(\eta)=\sum_{i=1}^n\sum_{j=1}^{\tau_{i}}b_{ij}(\eta)(Z_{ij}^2(\theta)-\sigma^2).
\end{equation}

We will often consider one of  the following conditions, which are progressively stronger.
\begin{condition}
\label{Condition2.1}
We assume that $\tau_i<\infty$ $a.s.$, $(i=1,\ldots).$
\end{condition}
We assume throughout that Condition \ref{Condition2.1} holds.
\begin{condition}
\label{Condition2.2}
We assume that $E[\tau_i]< \infty\; (i=1,\ldots).$
\end{condition}
Note that Condition \ref{Condition2.2} implies Condition \ref{Condition2.1} because $\tau_i>0,(i=1,\ldots).$
\begin{condition}
\label{Condition2.3}
$(T1)$  We assume that $\tau_{i}$ is  $a.s.$ bounded from above by a nonrandom integer $ m_{i} \;(i=1,\ldots)$.
\end{condition}
\begin{condition}
\label{Condition2.4}
  $(T2)$ There exists a non-random integer $m,$ such that $\tau_i\leq m\;(i=1,\ldots)$ $a.s$.
\end{condition}

We introduce more notation. Let
\begin{equation}
g_i(\theta)= \sup_{j}\|f_{ij}(\theta)Z_{ij}(\theta)\|,\;\;\;\;\; h_{i}(\eta)=\sup_{j}|b_{ij}(\eta)\\(Z_{ij}^2(\theta)-\sigma^2)|, (i=1,\ldots).\nonumber
\end{equation}

\begin{condition}
\label{Condition2.5}
The random variables $g_i(\theta)\tau_i$ and $h_i(\eta)\tau_i$ are integrable \\ $(i=1,\ldots,)$.
\end{condition}
Theorem 2.1 stated in \cite{clement and Strawderman} is proved in the Supplementary material from \cite{clement and Strawderman}.
In the first part of our Proposition \ref{Proposition1.1} we prove a more general result.
\begin{proposition} \label{Proposition1.1} Assume that Condition \ref{Condition2.5} holds.  Then the estimating functions in (\ref{sequation.2.8newf}--\ref{sequtaion.2.9.newf}) 
 are unbiased (\textit{i.e.}, have zero expectation). Furthermore, $a.s.$,
\begin{equation}
\label{equation1.13}
\qquad \qquad \qquad \qquad \qquad n^{-1} g_{n,k}(\eta_0)\rightarrow 0, \quad n\rightarrow \infty, \quad \quad \quad\quad(k=1,2).
\end{equation}
\end{proposition}

\begin{proof}

In conjunction with (\ref{equation1.9}--\ref{equation1.10}), we fix an index $i$ and omit writing it for now, along with the parameters $\theta$ and $\sigma$. With $(l=1,\ldots,)$, we
define the functions

\begin{equation}
\label{equation2.11.newf}
g_{1}^{[l]}= \sum_{j=1}^l f_jZ_jI\{S_{j-1}< C\}
\end{equation}
\begin{equation}
\label{equation2.12.newf}
g_{2}^{[l]}=\sum_{j=1}^l b_j(Z_{j}^2-\sigma^2) I\{S_{j-1}<C\}
\end{equation}
 Since $S_{l-1}\geq S_{\tau}\geq C,$ if $l>\tau$, $g_{k}^{[\tau]}$ is $g_{k}^{[l]}$ when $l\geq \tau\;(k=1, 2)$.
Now,
\begin{equation}
E[f_jZ_jI\{S_{j-1}<C\}|\mathcal{G}_{j-1}]=E[f_jI\{S_{j-1}<C\}E[Z_j|\mathcal{G}_{j-1}]]=0, \quad (j=1,\ldots,).\nonumber
\end{equation}
Similarly, $E[b_j(Z_{j}^2-\sigma^2)I\{S_{j-1}< C\}]=0$, by the definition of $\mathcal{G}_{j-1}$, $(A)$ and (\ref{euation1.2}). It follows that $g_{k}^{[l]}$ is a zero-mean martingale in $l$, and we also have

\begin{equation}
\label{equation1.14}
\qquad \qquad \qquad \quad \quad\quad  g_{k}^{[l]}\quad \rightarrow \quad g_{k}^{[\tau]}\quad\quad \quad  \mbox{a.s.}, \quad l\rightarrow\infty \quad \quad \;(k=1,2).
\end{equation}
Convergence in expectation follows if some form of uniform integrability is assumed (see \cite{Billingsley:1995}, p.464). Here $\displaystyle\|g_{1}^{[l]}\|\leq g\sum_{j=1}^{\infty}I\{S_{j-1}< C\}=g\tau$,
and $|g_{2}^{[l]}|\leq h\tau$. By Condition \ref{Condition2.5} with $i$ suppressed, the bounded convergence theorem  and (\ref{equation1.14})
\begin{equation}
\label{equation1.15}
\quad\qquad \qquad 0=E[g_{k}^{[l]}]\;\rightarrow \; E[g_{k}^{[\tau]}], \quad \quad  l\rightarrow \infty,\quad\quad (k=1,2)\nonumber
\end{equation}
To prove (\ref{equation1.13}), we write  $\displaystyle n^{-1} g_{n,k}(\eta_0)= n^{-1}\sum_{i=1}^n g_{k}^{[\tau_i]}(\eta_0)$, which converges to 0 $a.s., (k=1,2)$, by the strong law of large numbers (SLLN) for zero mean $i.i.d.$ variables.
\end{proof}

We introduce some notation. Let $i,j=1,\ldots,$ consider the set
\begin{equation}
 \{S_{i,j-1}<C_{i}<S_{ij}\}=\{0<C_i-S_{i,j-1}<S_{ij}-S_{i,j-1}\}\nonumber
\end{equation}

and note that it belongs to $\mathcal{G}_{ij}$.
The complement of this set is
\begin{equation}
\{C_i\leq S_{i,j-1}\}\cup \{S_{ij}\leq C_{i}\}= \{C_i-S_{i,j-1}\leq 0\}\cup\{S_{ij}-S_{i,j-1}\leq C_i-S_{i,j-1}\},\nonumber
\end{equation}
 $i.e.$,
it is a union of two disjoint sets. We define the set indicators
\begin{eqnarray}
I_{ij}^{obs}&=&I\{S_{ij}\leq C_{i}\},\quad\quad \quad I_{ij}^{cen}=I\{S_{i,j-1}< C_i< S_{ij}\},\nonumber\\
 I_{ij}^{out}&=&I\{C_i\leq S_{i,j-1}\}\;\;\quad(i,j=1,\ldots).\nonumber
\end{eqnarray}

Here ``obs'' stands for fully observed gap times, ``cen'' for censored gap times,
and ``out'' for fully unobserved gap times. Note that $I_{ij}^{out}$ is $\mathcal{G}_{i,j-1}$-measurable and the other two  set indicators are $\mathcal{G}_{ij}$-measurable $(i,j=1,\ldots)$.
Since
\begin{equation}
 \{S_{i,j-1}<C_{i}\}=\{S_{ij}\leq C_{i}\}\bigcup\{S_{i,j-1}< C_{i}<S_{ij}\},\nonumber
\end{equation}
 the estimating functions in (\ref{equation1.9}--\ref{equation1.10}) are, for $n=1,\ldots,$
\begin{eqnarray}
\label{equation1.16}
g_{n,1}(\theta)=\sum_{i=1}^n\sum_{j=1}^{\infty}f_{ij}(\theta)(Z_{ij}(\theta)I_{ij}^{obs}+Z_{ij}(\theta)I_{ij}^{cen}),
\end{eqnarray}
\begin{equation}
\label{equation1.17}
g_{n,2}(\eta)=\sum_{i=1}^n\sum_{j=1}^{\infty}b_{ij}(\eta)(Z_{ij}^2(\theta)-\sigma^2)(I_{ij}^{obs}+I_{ij}^{cen}),
\end{equation}
 It is the terms restricted to $I_{ij}^{cen}\;(i,j=1,\ldots)$ which will be imputed, using the observed data.

\subsection{Observed estimating functions}

In this section we describe a three-step procedure for imputing and estimating the censored terms of (\ref{equation1.16}--\ref{equation1.17}). We also discuss the unbiasedness properties of the resulting  estimating functions.  We start with (\ref{equation1.16}).

 We start  Step \ref{step2.1}  by discarding terms with $E[I_{ij}^{cen}]=0$ and writing the observed and empirical estimating functions (\ref{eauation.1.18}) and (\ref{equation.1.19}), respectively.
\begin{equation}
\label{eauation.1.18}
g_{n,1}^{obs}(\theta)=\sum_{i=1}^n\sum_{j=1}^{\tau_i}\bigg (f_{ij}(\theta)Z_{ij}(\theta)I_{ij}^{obs}+f_{ij}(\theta)I_{ij}^{cen}\frac{E_{\theta}[Z_{ij}(\theta)I_{ij}^{cen}]}{E_{\theta}[I_{ij}^{cen}]}\bigg).
\end{equation}
\begin{equation}
\label{equation.1.19}
\hat{g}_{n,1}^{obs}(\theta)=\sum_{i=1}^n\sum_{j=1}^{\tau_i}\bigg(f_{ij}(\theta)Z_{ij}(\theta)I_{ij}^{obs}-f_{ij}(\theta)I_{ij}^{cen}\frac{\sum_{k=1}^nZ_{kj}(\theta)I_{kj}^{obs} }{\sum_{k=1}^nI_{kj}^{cen}}\bigg).
\end{equation}
The empirical estimating functions $\hat{g}_{n,1}^{obs}(\theta)$ are used in our simulations to obtain estimators of $\theta_0$.

For estimating $\sigma_{0}^2$, we consider the following estimating functions
\begin{equation}
\label{equation.1.20}
g_{n,2}^{obs}(\eta)=\sum_{i=1}^n\sum_{j=1 }^{\tau_i}\bigg[b_{ij}(\eta)(Z_{ij}^2(\theta)-\sigma^2)I_{ij}^{obs}+b_{ij}(\eta)I_{ij}^{cen}\frac{E_{\eta}[(Z_{ij}^2(\theta)-\sigma^2)I_{ij}^{cen}]}
{E_{\eta}[I_{ij}^{cen}]}\bigg],
\end{equation}
\begin{equation}
\label{equation.1.21}
\hat{g}_{n,2}^{obs}(\eta)=\sum_{i=1}^n\sum_{j=1 }^{\tau_i}b_{ij}(\eta)\bigg[(Z_{ij}^2(\theta)-\sigma^2)I_{ij}^{obs}-I_{ij}^{cen}\frac{\sum_{k=1}^n(Z_{kj}^2(\theta)-\sigma^2)I_{kj}^{obs}}
{\sum_{k=1,k\neq i}^n I_{kj}^{cen}}\bigg].
\end{equation}

\begin{step}\label{step2.1} Replace $Z_{ij}(\theta)I_{ij}^{cen}$ with  $I_{ij}^{cen}E_{\theta}[Z_{ij}(\theta)I_{ij}^{cen}]/E_{\theta}[I_{ij}^{cen}]$.  Give conditions\\\\ under which the resulting $g_{n,1}^{obs}(\theta)$ in (\ref{eauation.1.18})  are unbiased. Adjust and apply the procedure to (\ref{equation1.17}).
\end{step}
 \begin{step}\label{step2.2}Show that $E_{\theta}[Z_{ij}(\theta)I_{ij}^{cen}]=-E_{\theta}[Z_{ij}(\theta)I_{ij}^{obs}]\;( i,j=1,\ldots),$ which justifies the definition of  $\hat{g}_{n,1}^{obs}(\theta)$ in (\ref{equation.1.19}). Apply the procedure to (\ref{equation1.17}).
\end{step}
\begin{step}\label{step2.3} Show that $g_{n,k}^{obs}(\theta)$  and $\hat{g}_{n,k}^{obs}(\theta)$, normalized by $n$, $(k=1,2)$ are asymptotically equivalent.
\end{step}
\begin{remark}
The justification for Step \ref{step2.1} is given in Section 6.1, where we show that the estimating functions (\ref{eauation.1.18}) and (\ref{equation.1.20}) are the projection of (\ref{equation1.16}), (\ref{equation1.17}), respectively, on an  appropriately defined $\sigma$-field.
\end{remark}

When $b_{ij}(\eta)=1$ for all $\eta\in K, i,j=1,\ldots,$ one can explicitly produce an estimator of $\sigma_{0}^2$ using  (\ref{equation.1.20}--\ref{equation.1.21}). The consistency of this estimator is easier to prove than the consistency of the implicitly defined estimators required in the general case.

To prove the unbiasedness of (\ref{eauation.1.18}), we define the $\sigma$-fields
\begin{equation}
\label{equation1.22}
\mathcal{O}(f_{i}(\theta))=\sigma\bigg(f_{ij}(\theta)I_{ij}^{cen},I_{ij}^{cen}, (j=1,\ldots )\bigg)\quad (i=1,\ldots).
\end{equation}
We recall that $I_{ij}^{cen}I_{ij'}^{cen}=0$ if $j\neq j'.$

\begin{condition}\;$(B_{f(\theta)})$:
\begin{equation}
\label{condtion.4}
(B_{f(\theta)})\qquad \qquad I_{ij}^{cen}E_{\theta}[Z_{ij}(\theta)|\mathcal{O}(f_{i}(\theta))]=I_{ij}^{cen}E_{\theta}[Z_{ij}(\theta)|I_{ij}^{cen}] \quad(i,j=1,\ldots).\nonumber
\end{equation}
\end{condition}
Condition $(B_{b(\eta)})$ is defined similarly, once we define $\mathcal{O}(b_{i}(\eta))$ as in (\ref{equation1.22}) by replacing $f_{ij}(\theta)$ with $b_{ij}(\eta)$.
We illustrate in Example \ref{example.1.5new} below  conditions $(B_{f(\theta)})$  and $(A)$, for $j=1,2.$  For simplicity, we introduce the notation, valid for $(i=1,\ldots,j=2,\ldots)$ :
\begin{equation}
\label{equation.final.2.3.8.New}
x_{ij,j}=\overline{BMI}_{ij}\;(j=1,\ldots), \; Y_{i}^{(j-1)}=\sum_{l=1}^{j-1}Y_{il}, \; x_{i}^{(j-1)}=\sum_{h=1}^{j-1}x_{ij,h}\;(j=2,\ldots).\nonumber
\end{equation}
Then in Example \ref{example1.4}
\begin{equation} \mu_{i1}(\theta)=\gamma_0+\gamma_1x_{i1,1},\;\mu_{ij}(\theta)=\gamma_{0}+\gamma_{1}x_{ij,j}+f_{j}(\rho)[Y_{i}^{(j-1)}-(j-1)\gamma_{0}-\gamma_{1}x_{i}^{(j-1)}]\nonumber
\end{equation}

\begin{example}\label{example.1.5new} Consider Example \ref{example1.4} with $V_{ij}(\theta)$ as in (\ref{equation1.5}). We first identify the generators of $\displaystyle \sigma(f_{ij}(\theta)), (i,j=1\ldots).$ Since $V_{ij}(\theta)$ is nonrandom, we only need to look at $\displaystyle \partial \mu_{ij}(\theta)/\partial \theta$. We have

\begin{equation}
\partial \mu_{i1}(\theta)/\partial \gamma_0=1;\;\; \partial \mu_{i1}(\theta)/\partial \gamma_1=x_{i1,1};\;\;\partial \mu_{i1}(\theta)/\partial \rho=0 \nonumber
\end{equation}
 and, with $(j=2,\ldots)$,

\begin{equation}
\frac{\partial \mu_{ij}(\theta)}{\partial \gamma_0}=1-(j-1)f_{j}(\rho),\quad
\frac{\partial \mu_{ij}(\theta)}{\partial \gamma_1}=x_{ij,j}-f_{j}(\rho)x_{i}^{(j-1)},\nonumber
\end{equation}
\begin{equation}
\frac{\partial \mu_{ij}(\theta)}{\partial \rho}=\dot{f}(\rho)[Y_{i}^{(j-1)}-(j-1)\gamma_0-\gamma_{1}x_{i}^{(j-1)}],\; \dot{f}_{j}(\rho)=[\rho(j-2)+1]^{-2}, by\; (\ref{equation1.3}).\nonumber
\end{equation}
\end{example}

 For $\displaystyle j=1, \mathcal{F}_{0}=\sigma(x_{i1})=\mathcal{G}_{i0}$, since $\{S_{i1}>0\}$ is the set of all outcomes. Thus $(A)$ holds in this case. On the other hand,

 \begin{equation}
 I_{i1}^{cen}E[Z_{i1}\mid \sigma(x_{i1}I_{i1}^{cen}, I_{i1}^{cen})]=I_{i1}^{cen}E[Z_{i1}\mid I_{i1}^{cen}], \nonumber
 \end{equation}
 by (\ref{equation1.22}) and $(B_{f(\theta)})$, which means that, in predicting $Z_{i1}$, the information provided by $x_{i1}I_{i1}^{cen}$ is irrelevant in the presence of the information that censoring has just occurred, $i.e.$, on the event $\{0<C_i<S_{i1}\}$.

 We recall that the covariate $x_{i1}$ was recorded at $S_0=0$. When $j=2$,
   \begin{equation}
   \mathcal{F}_{i1}=\sigma(x_{i1},x_{i2}, S_{i1})\;\;\mbox{ and}\;\; \mathcal{G}_{i1}=\sigma(\mathcal{F}_{i1},\{C_i>S_{i1}\},\{C_i=S_{i1}\})\nonumber
   \end{equation}
    so $(A)$ states that the complete history $\mathcal{F}_{i1}$ influences $Z_{i2}$, regardless of the information on the position of censoring relative to $S_{i1}$. Now $\displaystyle \sigma(\partial \mu_{i2}/ \partial\theta)$ is generated by $\displaystyle x_{i2,2}-f_2(\rho)x_{i2,1}$ and $\displaystyle Y_{i1}-\gamma_{1}x_{i2,1}$, both $\mathcal{F}_{i1}$-measurable, so $\displaystyle \sigma(\partial \mu_{i2}/\partial \theta)\subset \mathcal{F}_{i1}$. According to $(B_{f(\theta)})$, for $j=2$ and in the event that censoring has just occurred, $i.e.$, on $\displaystyle I_{i2}^{cen}=\{S_{i1}<C_i<S_{i2}\}$, this latter information prevails over the incomplete information provided by $\displaystyle \sigma(f_{i2}(\theta)I_{i2}^{cen})$ in influencing the incompletely observed $Z_{i2}$. Note that $Z_{ij}$ is disregarded on $I_{ij}^{out}$. Thus, conditions $(A)$ and $(B_{f(\theta)})$ complement each other in that, if $C_i$ has ``just'' occurred, this event, $i.e.$, $I_{ij}^{cen}=1$, is quite informative for $Z_{ij}$, unlike events of the type
    $ \{S_{ik}\leq C_i\}, \; k\leq j-1$.

\begin{proposition}\label{proposition1.2} Assume that Condition \ref{Condition2.5}, $(B_{f(\theta)})$ and $(B_{b(\eta)})$ hold. Then $g_{n,k}^{obs}(\theta)$
is unbiased, $k=1,2.$
\end{proposition}
\begin{proof} Case $k=1$. By (\ref{eauation.1.18}) and Proposition \ref{Proposition1.1}, we must show that, for $i,j=1,\ldots,$
\begin{equation}
E_{\theta}\bigg[f_{ij}(\theta)I_{ij}^{cen}\frac{E_{\theta}[Z_{ij}(\theta)I_{ij}^{cen}]}{E_{\theta}[I_{ij}^{cen}]}\bigg]=E_{\theta}\bigg[f_{ij}(\theta)I_{ij}^{cen}Z_{ij}(\theta)\bigg].\nonumber
\end{equation}
The left hand side is
\begin{eqnarray}
\label{qua.ass.b3}
E_{\theta}\bigg[f_{ij}(\theta)I_{ij}^{cen}E_{\theta}[Z_{ij}(\theta)|I_{ij}^{cen}]\bigg]&=&E_{\theta}\bigg[f_{ij}(\theta)I_{ij}^{cen}E_{\theta}[Z_{ij}(\theta)|\mathcal{O}(f_{i}(\theta))]\bigg]\nonumber\\
&=&E_{\theta}\bigg[f_{ij}(\theta)I_{ij}^{cen}Z_{ij}(\theta)\bigg],\nonumber
\end{eqnarray}
by $(B_{f(\theta)})$ and (\ref{equation1.22}). The proof for $k=2$ is similar, with $(B_{f(\theta)})$ replaced by $(B_{b(\eta)})$.
\end{proof}

We now proceed with Step \ref{step2.2}, embodied in the following lemma.
\begin{lemma}\label{Lemma.1} Let $g_{ij}(\theta)$ be  $\mathcal{G}_{i,j-1}$-measurable, $(i,j=1,\ldots)$. We have, when all integrals exist:
\begin{equation}
\label{equationnew21}
E_{\eta}[g_{ij}(\eta)Z_{ij}(\theta)I_{ij}^{cen}|\mathcal{G}_{i,j-1}]=-E_{\eta}[g_{ij}(\eta)Z_{ij}(\theta)I_{ij}^{obs}|\mathcal{G}_{i,j-1}].
\end{equation}
\begin{equation}
\label{equationnew22}
E_{\eta}[g_{ij}(\eta)(Z_{ij}^{2}(\theta)-\sigma^2)I_{ij}^{cen}|\mathcal{G}_{i,j-1}]=-E_{\eta}[g_{ij}(\eta)(Z_{ij}^{2}(\theta)-\sigma^2)I_{ij}^{obs}|\mathcal{G}_{i,j-1}].
\end{equation}
\end{lemma}
 \begin{proof}

Since
\begin{eqnarray}
 0&=&E_{\eta}[g_{ij}(\eta)E_{\theta}[Z_{ij}(\theta)|\mathcal{G}_{i,j-1}]]
=E_{\eta}[g_{ij}(\eta)E_{\theta}[Z_{ij}(\theta)I_{ij}^{obs}|\mathcal{G}_{i,j-1}]]\nonumber\\
&+&E_{\eta}[g_{ij}(\eta)E_{\theta}[Z_{ij}(\theta)I_{ij}^{cen}|\mathcal{G}_{i,j-1}]]
+E_{\eta}[g_{ij}(\eta)E_{\theta}[Z_{ij}(\theta)I_{ij}^{out}|\mathcal{G}_{i,j-1}]],\nonumber
\end{eqnarray}
 by $(A)$, (\ref{euation1.2}) and $1=I_{ij}^{obs}+I_{ij}^{cen}+I_{ij}^{out}$. Since
 \begin{equation}
 E_{\theta}[Z_{ij}(\theta)I_{ij}^{out}|\mathcal{G}_{i,j-1}]=I\{S_{i,j-1}\geq C_{i}\}E_{\theta}[Z_{ij}(\theta)|\mathcal{G}_{i,j-1}]=0,\nonumber
 \end{equation}
by (A) and (\ref{euation1.2}), (\ref{equationnew21}) holds. We prove (\ref{equationnew22}) in a similar manner.
\end{proof}

We now continue with Step \ref{step2.3} and show the asymptotic equivalence of the normalized estimating functions $g_{n,k}^{obs}$ and $\hat{g}_{n,k}^{obs}, (k=1,2)$.

\begin{theorem}\label{theorem.1}Assume that Condition \ref{Condition2.5}, $(B_{f(\theta)})$, $(B_{b(\eta)})$ and $(T2)$ hold. Furthermore, assume that there exists $C>0$, such that, $a.s.$,
\begin{equation}
\label{equationnew23}
\max\bigg \{\sup_{i,j\geq 1}| b_{ij}(\eta_0)|,\sup_{i,j\geq 1}\parallel f_{ij}(\theta_0)\parallel\bigg \} \leq C<\infty.
\end{equation}
Then, $a.s.$,
\begin{equation}
\label{equationnew24}
n^{-1}\parallel g_{n,k}^{obs}(\theta_0)-\hat{g}_{n,k}^{obs}(\theta_0)\parallel\rightarrow 0,\quad
n^{-1}\parallel g_{n,k}(\theta_0)-g_{n,k}^{obs}(\theta_0)\parallel\rightarrow 0 \quad n\rightarrow\infty,
\end{equation}
and therefore, by (\ref{equation1.13}), $ n^{-1}\hat{g}_{n,k}^{obs}(\theta_0)\rightarrow 0 $ when $n\rightarrow\infty, (k=1,2)$.
\end{theorem}
\begin{proof} From (\ref{eauation.1.18}--\ref{equation.1.19}), to prove the first assertion in (\ref{equationnew24}), we evaluate, using (\ref{equationnew23}),
\begin{equation}
n^{-1}\parallel g_{n,1}^{obs}(\theta_0)-\hat{g}_{n,1}^{obs}(\theta_0)\parallel
\leq C \sum_{j=1}^{m}\bigg|\frac{E[Z_{1j}(\theta_0)I_{1j}^{obs}]}{E[I_{1j}^{cen}]}-\frac{\sum_{k=1}^{n}Z_{kj}(\theta_{0})I_{kj}^{obs}}{\sum_{k=1 }^{n}I_{kj}^{cen}}\bigg|,\nonumber
\end{equation}

For each $j$ we can apply the SLLN to each sum within the absolute value, and (\ref{equationnew24}) follows for $k=1$.
Now,\\
$ \displaystyle n^{-1}(g_{n,1}(\theta_0)-g_{n,1}^{obs}(\theta_0))=n^{-1}\sum_{i=1}^n\sum_{j=1}^{m}f_{ij}(\theta_0) I_{ij}^{cen} [Z_{ij}(\theta_0)-E[Z_{1j}(\theta_0)I_{1j}^{cen}]/E[I_{1j}^{cen}]
]
$ converges $a.s.$ to
$\displaystyle \sum_{j=1}^{m} [E[f_{1j}(\theta_0)I_{1j}^{cen}Z_{1j}(\theta_0)]-E[Z_{1j}(\theta_0)I_{1j}^{cen}]E[f_{1j}(\theta_0)I_{1j}^{cen}]/E[I_{1j}^{cen}]].$
By (\ref{equation1.22}) and $(B_{f(\theta)})$,
\begin{equation} E[f_{ij}(\theta_0)I_{ij}^{cen}Z_{ij}(\theta_0)|\mathcal{O}(f_{i}(\theta_0))]
=f_{ij}(\theta_0)I_{ij}^{cen}\\E[Z_{1j}(\theta_0)I_{1j}^{cen}]/E[I_{1j}^{cen}].\nonumber
\end{equation}
Taking expectation proves the second part of (\ref{equationnew24}), for $k=1$. The proof for $k=2$ is similar.
\end{proof}
\begin{remark}\label{remarknew2.3}
The results of Sections 2.2--2.3 are still valid if we replace assumption $(A)$ with the weaker assumption $(A0)$.
\end{remark}

\subsection{Connection with other longitudinal studies}

 We consider the left side of (\ref{equation0.2}) under the working independence assumption. The variances at the denominator in (\ref{equation0.2}) are the diagonal matrices
\begin{equation} V_{i}(\beta)=
\mbox{diag}_{1\leq j\leq m_i}\{ \sigma_{ij}(\beta) \}I_{m_i}\mbox{diag}_{1\leq j\leq m_i}\{\sigma_{ij}(\beta)\}=
\mbox{diag}_{1\leq j\leq m_i}\{\sigma_{ij}^2(\beta)\},\nonumber
\end{equation}

 where $I_{m_i}$ is the $m_i\times m_i$ identity matrix and the diagonal entries $\sigma_{ij}(\beta)$  correspond to $V_{ij}(\theta)$ in (\ref{euation1.2}), when $\theta=\beta$ and the conditioning $\sigma$-fields in (\ref{euation1.2}) contain only the entire set of outcomes  and the empty set. Then (\ref{equation0.2})  becomes:
\begin{equation}
\sum_{i=1}^n\bigg(\frac{\partial \mu_{i}(\beta)}{\partial \beta^T}\bigg)^T \mbox{diag}_{1\leq j\leq m_i}\{ \sigma_{ij}^{-2}(\beta)\}
(y_i-\mu_{i}(\beta))=0.\nonumber
\end{equation}
 When (T1) holds, (\ref{equation1.9}) becomes
\begin{equation}
\label{eqnuation.2.5.second}
\sum_{i=1}^{n}\frac{\partial \mu_{i}(\theta)}{\partial \theta}\mbox{diag}_{1\leq j\leq m_i}\{ V_{ij}^{-2}(\theta)\}
(y_i-\mu_{i}(\theta))=0.\nonumber
\end{equation}
Since $\partial \mu_{i}(\theta)/\partial \theta=(\partial \mu_{i}(\theta)/\partial \theta^T)^T$, the similarity of (\ref{equation0.2}) and (\ref{equation1.9}) is now apparent. In this context and in
a longitudinal study, each individual $i$ is observed at random times $S_{ij}, 1\leq j\leq \tau_i,i\geq 1$. The random covariates $x_{ij}$ are available at time $S_{i,j-1}$, while the response variable $y_{ij}$, which satisfies (\ref{euation1.2}), is recorded at time $S_{ij}$.  It could represent, for instance, the result of a blood test, or some characteristic that may require some costly effort to obtain. In this case the recurrent events themselves are the object of the analysis, rather than the gap times. In a study of asthma in children we could be interested in the intensity, duration or type of the asthma episodes, rather than the time gaps between episodes. Note that in this case only the estimating functions (\ref{equation1.9}--\ref{equation1.10}) are needed.

While $\tau_i$ remains as in Definition \ref{definitionS2.1}, the response variable $y_{ij}$ replaces $Y_{ij}$ in (\ref{euation1.2}), (\ref{equation1.7}) and beyond. Therefore our method can be applied to more general longitudinal studies, $e.g.$ when data is collected at random times and censoring occurs. This approach also provides estimators for the overdispersion parameter.

\section{Strong consistency}\label{strong}

\subsection{General results}
We first introduce the necessary definitions.
The numerical radius of a $p\times p$ matrix $A$ is $|||A|||=\sup_{\parallel \lambda\parallel=1}|\lambda^TA\lambda|,\;\lambda\in R^p.$
The Euclidian norm is denoted $\parallel A\parallel $. These  norms are asymptotically equivalent, and, depending on the situation, one can use the most convenient one to prove asymptotic results.

Let $ B_{r}(\theta)=\{\theta'\in R^p: \parallel \theta'-\theta\parallel\leq r\}$.
Let $q_{n}(\eta)=\sum_{i=1}^n u_{i}(\eta)$, where $\eta\in K \subset R^{p+1}$ is a parameter, $u_{i}(\eta)\in R^{p+1}$ are random vectors, which are square integrable and continuously differentiable in $\eta$. Let $\mathcal{D}_{n}(\eta)=-\partial q_{n}(\eta)/(\partial \eta^T)$ be the $(p+1)\times (p+1)$ matrix of derivatives.
The following theorem gives sufficient conditions for the almost sure existence and strong consistency of a sequence of estimators of $\eta_0$ (see Theorem 4.2 of \cite{Balan.etal:2010}).
\begin{theorem}\label{theorem2} Assume that the following conditions hold $a.s.$,

\item[(LN)]  $n^{-1}q_{n}(\eta_0)\rightarrow 0,$   when $ n\rightarrow \infty, $

\item[(S)]  There exist random variables $C_{0}>0$, $r_1>0$, and a random integer $n_1\geq 1$, such that,
 for all $\lambda \in R^{p+1},\parallel\lambda\parallel=1$,\\\\
$(i)$    $\quad  \inf\limits_{n\geq n_1}\inf\limits_{\eta \in B_{r_1}(\eta_0)}|\lambda^T\mathcal{D}_{n}(\eta)\lambda|>0;$\\\\
$(ii)$  $\quad \lim\limits_{r\rightarrow 0}\limsup\limits_{n\rightarrow\infty}\sup\limits_{\eta \in B_{r}(\eta_0)}n^{-1}|||\mathcal{D}_{n}(\eta)-\mathcal{D}_{n}(\eta_{0})|||=0 ;$\\\\
$(iii)$ $\quad \inf\limits_{n\geq n_1} n^{-1}|\lambda^T\mathcal{D}_{n}(\eta_0)\lambda|\geq C_{0}.$\\
Then, there exists a sequence of random vectors $\{\hat{\eta}_{n}\}\subset R^{p+1}$,
and a random integer $n_{0}$, such that:\\
$ P\{q_{n}(\hat\eta_{n})=0,\quad \mbox{for all}\quad n\geq n_{0}\}=1$ and
 $ \hat\eta_{n}\rightarrow\eta_{0} $ $a.s.$,  as  $n\rightarrow \infty.$
\end{theorem}
\begin{proof} The proof is identical to the proof of Theorem 4.2 in \cite{Balan.etal:2010}, once $\alpha_{n}^{1/2+\delta}$ is replaced by $n$. Conditions $S(i)$ and $S(iii)$ ensure nonsingularity and $S(ii)$ the equicontinuity of the derivatives at $\eta_0$.
\end{proof}

 \begin{remark}\label{remark.2}  Condition $(LN)$ replaces the unbiasedness of $q_{n}(\eta)$ at $\eta_0,n\geq 1$. It is satisfied if we can find an estimating function $q_{n,0}(\eta)$, for which $(LN)$ holds at $\eta_0$, and such that,
$
 n^{-1}(q_{n}(\eta_0)-q_{n,0}(\eta_0))\rightarrow 0$ $a.s.$ as $ n\rightarrow \infty.$
To apply Theorem \ref{theorem2}, we could take $q_n(\eta)=(\hat{g}_{n,1}^{obs}(\eta),\hat{g}_{n,2}^{obs}(\eta))$ from (\ref{equation.1.19}), (\ref{equation.1.21}) and $q_{n,0}(\eta)=(g_{n,1}^{obs}(\eta),g_{n,2}^{obs}(\eta))$ from (\ref{eauation.1.18}), (\ref{equation.1.20}) respectively, to obtain the  convergence results above and $\hat{\eta}_{n}$, such that $q_{n}(\hat{\eta}_{n})=0$  for $ n\geq n_0$ and $\hat\eta_n\rightarrow \eta_0$ $a.s.$
To check condition $(LN)$, we can use (\ref{equationnew24}) from Theorem \ref{theorem.1} and the results of Proposition \ref{Proposition1.1}.
\end{remark}

 We first find a specific sequence of estimators $\hat{\eta}_{n}$ of $\eta_0$, which is consistent and has a well-defined asymptotic distribution. This sequence is obtained in three steps.

\begin{step}  Solve the system $\hat{g}_{n,1}^{obs}(\theta)=0$, to obtain $\hat{\theta}_{n},$ with $ \hat{\theta}_{n}\rightarrow \theta_0$ $a.s.$
\end{step}

\begin{step}  Solve $\hat{g}_{n,2}^{obs}(\hat{\theta}_{n}, \sigma^2)=0$, to  obtain
$\hat{\sigma}^{2}_{n}=\sigma^{2}_{n}(\hat{\theta}_{n})\rightarrow \sigma^{2}_{0}$ $a.s.$
\end{step}

 \begin{step}Put together  $(\hat{\theta}_{n},\hat{\sigma}_{n}^{2})$, to obtain a sequence of strongly  consistent estimators of $\eta_0$.
  \end{step}

 If only the main regression parameter $\theta$ is of interest, then one can deal with a simpler, self-contained version of Theorem \ref{theorem2}, while $\sigma^2$, along with $\hat{g}_{n,2}^{obs}(\eta)$, can be completely ignored.

 In order to apply Theorem \ref{theorem2}, conditions on the moduli of continuity of functions related to the derivative in Theorem \ref{theorem2} are needed. In \cite{Balan.etal:2010}, only the analytical properties of $\mu$ are needed, as $V_{ij}(\theta)$ is a function of $\mu_{ij}(\theta)$ and $c_{ij}(\theta)$ is a linear function there. As for $b_{ij}(\eta)$, it is not present since the overdispersion parameter is not considered there.

\subsection{Asymptotic results for derivatives}
Let
$\displaystyle \quad \;\;h_{ij}(\theta)=\dot\mu(c_{ij}^{T}(\theta)x_{ij})\partial c_{ij}^T(\theta)/{\partial \theta},\;
 \mbox{or}\;\;  h=\dot{\mu}\partial{c^T}/\partial{\theta},$\\\\

$\displaystyle \quad \quad \quad \delta_{r,n}(h)=\sup_{\theta',\theta\in B_{r}(\theta_0)}\max_{1\leq i\leq n,1\leq j\leq \tau_i}V_{ij}^{-1}(\theta')\parallel
 h_{ij}(\theta)-h_{ij}(\theta_0)\parallel,\;$\\\\

$\displaystyle\quad \quad \quad b_{r,n}(h)=\sup_{\theta', \theta \in B_{r}(\theta_0)}\max_{1\leq i\leq n, 1\leq j\leq \tau_i}V_{ij}^{-1}(\theta')\parallel h_{ij}(\theta)\parallel.$ \\ \\

Let $X_i$ denote the $\tau_i\times q$ matrix of covariates, with $jh$ entry $(X_i)_{jh}=x_{ij,h}, (j=1,\ldots,\tau_i; h=1,\ldots,q)$ and $(i=1,\ldots),$ where $q$ is the maximum number of covariates over  all individuals, assumed to be finite. Note that \\$\displaystyle \sum_{j=1}^{\tau_i}\parallel x_{ij}\parallel^2=tr(X_{i}^TX_i)$. Let $\displaystyle H_n(\theta)=\sum_{i=1}^n\sum_{j=1}^{\tau_i}f_{ij}(\theta)f_{ij}^T(\theta)I_{ij}^{obs}$. The following result illustrates condition $S(ii)$ on the leading term in the decomposition of $\mathcal{D}_{n}(\eta)$ (see \cite{Balan.etal:2010}).
\begin{lemma}
\label{lemma2}
Assume that
 \begin{equation}E[tr(X_{1}^TX_1)]<\infty,\;\;\lim_{r\rightarrow 0}\limsup_{n\rightarrow \infty}\delta_{r,n}(h)=0 \;\mbox{and}\; \;\lim_{r\rightarrow 0}\limsup_{n\rightarrow \infty}b_{r,n}(h)<\infty.\nonumber
 \end{equation}
  Then
\begin{equation}
\lim_{r\rightarrow 0}\limsup_{n\rightarrow \infty} \sup_{\theta \in B_{r}(\theta_0)}n^{-1}|||H_{n}(\theta)-H_{n}(\theta_0)|||=0.\nonumber
\end{equation}
\end{lemma}
\begin{proof}
Let $\theta\in B_{r}(\theta_0)$. Using the Cauchy-Schwarz inequality for each term, with $\displaystyle\lambda\in R^p,$ $ \displaystyle \parallel \lambda\parallel=1,$
\begin{eqnarray}
   |||H_{n}(\theta)-H_{n}(\theta_0)|||&\leq& \sum_{i=1}^n\sum_{j=1}^{\tau_i}\lambda^T(f_{ij}(\theta)-f_{ij}(\theta_0))(f_{ij}(\theta)-f_{ij}(\theta_0))^T\lambda I_{ij}^{obs}\nonumber\\
   &+&2\mid \sum_{i=1}^n\sum_{j=1}^{\tau_i}\lambda^T f_{ij}(\theta_0)(f_{ij}(\theta)-f_{ij}(\theta_0))^T\lambda I_{ij}^{obs}\mid\nonumber\\
  &\leq &\delta_{r,n}^2(h)\sum_{i=1}^n\sum_{j=1}^{\tau_i}\lambda^T x_{ij}x_{ij}^T\lambda\nonumber\\
  &+&2\sum_{i=1}^n\sum_{j=1}^{\tau_i}
   b_{r,n}(h)\delta_{r,n}(h)\parallel x_{ij}\parallel^2\nonumber\\
   &\leq &\delta_{r,n}(h)\sum_{i=1}^n tr(X_{i}^T X_i)(\delta_{r,n}(h)+2 b_{r,n}(h)).\nonumber
\end{eqnarray}
The conclusion follows from the hypotheses and the SLLN.
\end{proof}
We discuss the conditions of Lemma \ref{lemma2} at the end of  Appendix \ref{app}.

\section{The asymptotic normality of estimators}\label{asyptotic}
\subsection{The asymptotic normality of $\hat{\theta}_n$}
In this section we present three results. The first shows that $n^{-1/2}g_{n,1}^{obs}(\theta_0)$ is asymptotically normal with mean zero and a nonsingular covariance matrix $\Sigma$, denoted $N(0,\Sigma)$. The second result shows that the asymptotic distribution of $n^{-1/2}\hat{g}_{n,1}^{obs}(\theta_0)$ is also $N(0,\Sigma)$. Finally, we show that $\sqrt{n}(\hat{\theta}_{n}-\theta_0)$ is asymptotically normal with mean 0 and a covariance matrix of a sandwich form, $\mathcal{D}_{1}^{-1}(\theta_0)\Sigma(\mathcal{D}_{1}^{-1}(\theta_0))^T,$ with $\mathcal{D}_{1}(\theta_0)$ nonrandom and nonsingular, where, $a.s.$,
 \begin{equation}
\label{quation.4.1.1.0001}\mathcal{D}_{1}(\theta_0)=\lim_{n\rightarrow \infty}n^{-1}\mathcal{D}_{n,1}(\theta_0),\quad \mbox{ and}\quad \mathcal{D}_{n,1}(\theta_0)=-\frac{\partial \hat{g}_{n,1}^{obs}(\theta_0)}{\partial\theta^T},\quad n\geq1.
\end{equation}
Proofs of all stated results are in Appendix \ref{appb.2}.

We introduce the notations:
 \begin{equation} g_{n,1}^{obs}(\theta_{0})=\sum_{i=1}^{n}u_{i,1}, \quad u_{i,1}=\sum_{j=1}^{m}(a_{ij}^{(1)}I_{ij}^{obs}-b_{ij}^{(1)}I_{ij}^{cen}),\nonumber
 \end{equation}
where
 \begin{equation} a_{ij}^{(1)}=f_{ij}(\theta_0)Z_{ij}(\theta_0), \quad b_{ij}^{(1)}=f_{ij}(\theta_0)E_{\theta_0}[Z_{1j}(\theta_0)I_{1j}^{obs}]/
E_{\theta_0}[I_{1j}^{cen}].\nonumber
\end{equation}
Similarly, $\displaystyle g_{n,2}^{obs}(\eta_{0})=\sum_{i=1}^nu_{i,2}.$

In the course of proving the following result, we calculate all entries of $\Sigma$.
\begin{theorem}\label{equ.th.4.1.1} Assume that $(T2)$  and the assumptions of Proposition \ref{proposition1.2} hold, and that
\begin{equation}
\label{equ.con.4.1.1.f}
\max_{1\leq j\leq m}E_{\theta_0}\bigg[\parallel f_{1j}(\theta_0)\parallel^2Z_{1j}^{2}(\theta_0)\bigg]<\infty,\;\;
\max_{1\leq j\leq m}E_{\theta_0}\bigg[\parallel f_{1j}(\theta_0)\parallel^2\bigg]<\infty.
\end{equation}
  Then
$\displaystyle
n^{-1/2}g_{n,1}^{obs}(\theta_{0})\;\rightarrow \; N(0, \Sigma)\;\mbox{ in distribution, as}\;\; n\rightarrow \infty,\nonumber
$
where $N(0, \Sigma)$ is a $p$-dimensional random vector normally distributed with zero mean and covariance matrix $\Sigma$. The $(kl)$-entry of $\Sigma$ is:

\begin{eqnarray}
\label{eqn.n.4.1.4.1.26}
&&\sum_{j,j'=1}^{m}E_{\theta_0}[a_{1j,k}^{(1)}a_{1j',l}^{(1)}I_{\max\{j,j'\}}^{obs}]\nonumber
-\sum_{j=1,j'>j}^{m}E_{\theta_0}[a_{1j,k}^{(1)}I_{1j'}^{cen}b_{1j',l}^{(1)}]\nonumber\\
 &-&\sum_{j'=1,j>j'}^m E_{\theta_0}[a_{1j',l}^{(1)}I_{1j}^{cen}b_{1j,k}^{(1)}]+\sum_{j=1}^{m}E_{\theta_0}[b_{1j,k}^{(1)}b_{1j,l}^{(1)}I_{1j}^{cen}].
\end{eqnarray}
\end{theorem}

\begin{proposition}\label{Pro.4.1.2} Under the conditions of Theorem \ref{equ.th.4.1.1}, $n^{-1/2}\hat{g}_{n,1}^{obs}(\theta_{0})$  and  $n^{-1/2}g_{n,1}^{obs}(\theta_{0}) $   have the same asymptotic distribution $N(0,\Sigma)$.
\end{proposition}
\begin{theorem}\label{theorem.clt.4.13}
 Assume that the conditions of theorems \ref{theorem2} and \ref{equ.th.4.1.1} hold,  with $q_{n}(\eta)$ and $\mathcal{D}_{n}(\eta)$ replaced by $\hat{g}_{n,1}^{obs}(\theta)$ and $\mathcal{D}_{n,1}(\theta)$, respectively . Assume further that condition (\ref{quation.4.1.1.0001}) holds, with $\mathcal{D}_{1}(\theta_0)$  nonrandom and invertible. Then

\begin{equation}
n^{1/2}(\hat{\theta}_{n}-\theta_0)\rightarrow N(0,\mathcal{D}_{1}^{-1}(\theta_0)\Sigma(\mathcal{D}_{1}^{-1}(\theta_0))^T).\nonumber
\end{equation}

\end{theorem}

\begin{remark}
The convergence of the normalized derivative in the hypotheses of Theorem \ref{theorem.clt.4.13} justifies the use of condition $(S)(iii)$ in Theorem \ref{theorem2}.
\end{remark}

\subsection{The asymptotic normality of $\hat{\sigma}_{n}^2$}

The presentation in this section is similar to that of the previous section. We detail only the presentation of the main results. The proofs are in Appendix \ref{appb.2}.

We write:
$\displaystyle \quad \hat{g}_{n,1}^{obs}(\theta)=\sum\limits_{i=1}^{n}\hat{u}_{i,1}(\theta), \;  \mbox{and} \; \hat{g}_{n,2}^{obs}(\theta,\sigma^2)=\sum\limits_{i=1}^{n}\hat{u}_{i,2}(\theta,\sigma^2)
$.\\

Let
$ \displaystyle \quad \mathcal{D}_{n,2}(\theta,\sigma^2)=-\partial \hat{g}_{n,2}^{obs}(\theta, \sigma^2)/\partial \sigma^2,\;
\mathcal{D}_{n,3}(\theta,\sigma^2)=-\partial \hat{g}_{n,2}^{obs}(\theta, \sigma^2)/\partial \theta^T,$\\\\
where $\displaystyle \mathcal{D}_{n,2}(\theta,\sigma^2) $ is a scalar and $ \displaystyle \mathcal{D}_{n,3}(\theta,\sigma^2)$ is of dimension $1\times p$.
\begin{theorem}\label{theroem5} Assume that the conditions of Theorem \ref{theorem.clt.4.13} hold. Let $\hat{\sigma}_{n}^2$ be a sequence of strongly consistent estimators of $\sigma_{0}^2$, such that $\displaystyle \hat{g}_{n,2}^{obs}(\hat{\theta}_{n},\hat{\sigma}_{n}^2)=0$ $a.s.$, for large $n$. Furthermore, assume that (\ref{equ.new.final.4.2.15}) holds $a.s.$:

\begin{equation}
 \label{equ.new.final.4.2.15}
\lim_{r\rightarrow 0}\limsup_{n\rightarrow\infty}n^{-1}\sup_{\eta \in B_{r}(\eta_0)}\bigg|\partial\hat{g}_{n,2}^{obs}(\eta)/\partial\eta^T-\partial\hat{g}_{n,2}^{obs}(\eta_0)/\partial\eta^T\bigg|=0.
\end{equation}
Assume that the sequence of derivatives converges $a.s.$ as follows:

\begin{equation}
\label{equ.new.final.4.2.16}
  n^{-1}\mathcal{D}_{n,2}(\eta_0)\rightarrow \mathcal{D}_{2}(\eta_0)\neq 0,\quad\quad \quad
n^{-1}\mathcal{D}_{n,3}(\eta_0)\rightarrow \mathcal{D}_{3}(\eta_0),
\end{equation}
with $\mathcal{D}_{2}(\eta_0)$ and $\mathcal{D}_{3}(\eta_0)$ nonrandom. Then
\begin{equation}
n^{1/2}(\hat{\sigma}_{n}^{2}-\sigma_{0}^2)\rightarrow N(0,\mathcal{D}_{2}^{-2}(\eta_0)\Omega),\nonumber
\end{equation}
where
$
\Omega=E[Q^{2}(\theta_0,\sigma_{0}^{2})],\; Q(\theta_0,\sigma_{0}^{2})=u_{1,2}-\mathcal{D}_{3}(\eta_0)\mathcal{D}_{1}^{-1}(\theta_0) u_{1,1}.\nonumber
$
\end{theorem}

\begin{remark}The results of Theorem \ref{theroem5} are similar to the results in 3.5 of \cite{Yi.et al: 2012}.
\end{remark}

\section{Simulation results}\label{simulation}

In this section, we investigate our proposed estimators $\hat{\eta}_n$ from (\ref{equation.1.19}) and  (\ref{equation.1.21}), with $b_{ij}(\eta)=1,(i,j=1,\ldots).$ We compare their performance to the conditional GEE estimators  in \cite{clement and Strawderman}, henceforth abbreviated C\&S  (available through the R package \texttt{condGEE}).

 We consider a {\em factorial design with 4 parameters}, two sample sizes $n=50$ and $n=200$  and two censoring schemes, 
 $C=125$ and $C=225$,  and four standardized distributions of the error terms: normal, exponential, uniform and log-normal.  For each combination of the design parameters, 1000 simulated studies are done to obtain parameter estimates, which are then compared in terms of bias and estimated standard error.
The simulated gap times (in days) are generated according to the model $ Y_{ij}=\max\{Y_{ij}^{*},1\}$, $j\geq 1$, $i=1,2\cdots,n$, where $ Y_{ij}^{*}=\mu_{ij}(\theta)+\sigma V_{ij}(\theta)\varepsilon_{ij} $ and $\varepsilon_{ij}$ are $i.i.d.$ observations from a density with mean zero and variance one. This translates to each of the following schemes: $N(0,1)$, shifted exponential with mean 0 and rate 1,  uniform on $[-\sqrt{3},\sqrt{3}]$, and lognormal $\exp{(X)}-\exp{(0.4812119)}$, where $X\sim N(0,\sqrt{0.4812119})$.

 We first compare our nonparametric method, abbreviated NP, to a correctly specified conditional model from  C\&S, as seen in Figures \ref{figbias1}-\ref{figese2}. Other schemes are presented in Table \ref{tablelabe2} and in the Supplementary material.

\begin{figure}
\includegraphics[width=.9\textwidth]{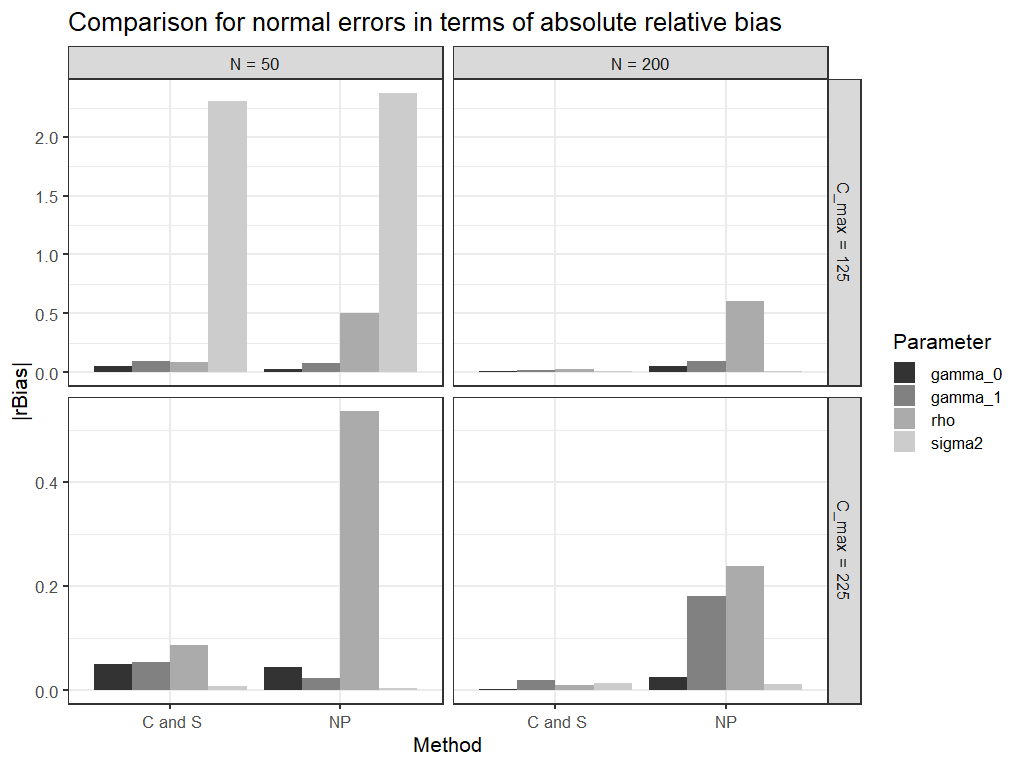}
\caption{Bias comparison}
\label{figbias1}
\end{figure}\par

As in \cite{Murphy.etal:1995}, the  conditional mean and variance functions specifications are given by (\ref{equation1.3})-(\ref{equation1.5}), with a slight difference: $\gamma_0$ in (\ref{equation1.4}) is $28+\gamma_0$ here. The parameters are:
 $\gamma_0=0.6,$ $ \gamma_1=-0.4 $ , $\rho=0.03, \sigma^2=11$.  A single time-varying covariate is used: $\overline{BMI}_{ij}=BMI_{ij}-21,$ where $BMI_{ij}$ is assumed to decrease linearly from 22 $kg/m^2$ on day 1 to 20 $kg/m^2$ on day 195, increase linearly to 21 $kg/m^2$ on day 225, and then remain constant thereafter.
With simulated observation periods of 125 and 225 days (roughly equal to 4 and 7.5 months, respectively), the average number of events per subject under an observation period is approximately 3.9 and 7.4, with little noise between subjects .

\begin{figure}
\includegraphics[width=.9\textwidth]{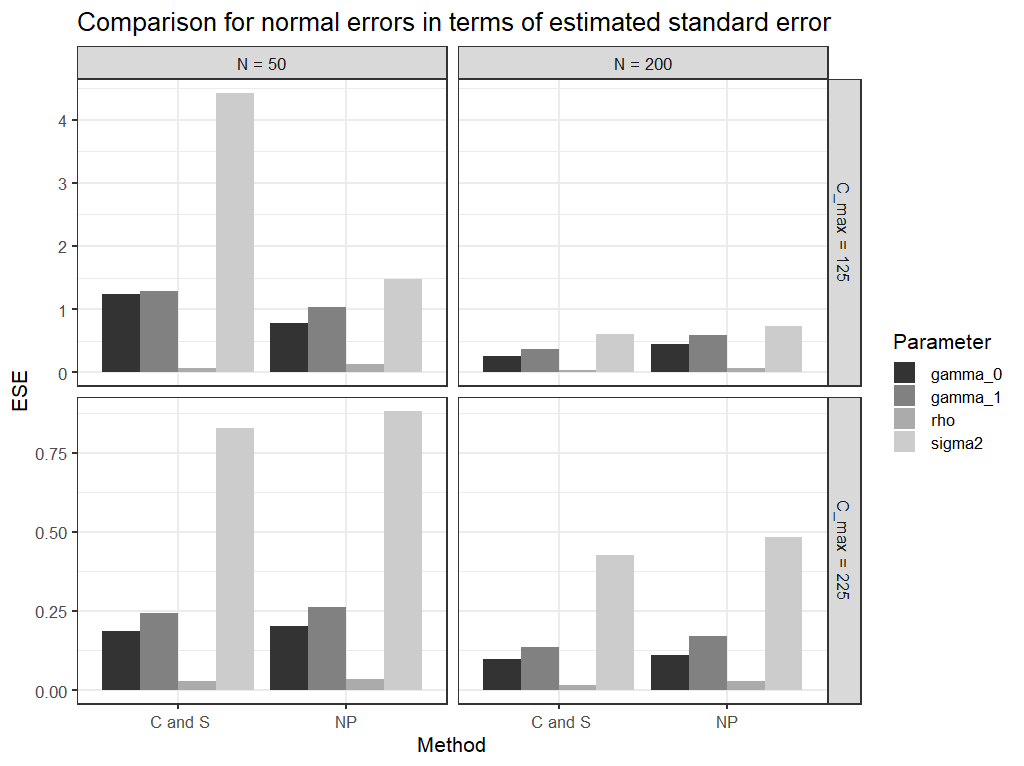}
\caption{Standard error comparison}
\label{figese2}
\end{figure}\par
Figures \ref{figbias1} and \ref{figese2} summarize results for bias and empirical standard error respectively. Our proposed method is quite biased in estimating $\rho$, which seems to also hold for other models in the Supplementary material. This may be due to the small value of $\rho$. On the other hand, our estimator of $\sigma^2$ behaves well for most scenarios. Our method seems to perform slightly better than C\&S for $\gamma_0$ and $\gamma_1$, especially on a shorter observation period with a small sample. As expected, on a large sample the method from C\&S, when correctly specified performs very well. The standard errors are similar for both methods, except on a small sample with a short observation period, where our method performs rather well. We note that the range of the $y$ axis depends on the value of $C_{\max}$ in the graphs.

As expected, increasing the sample size improves estimation for both methods. Our method is often better suited for small sample size and shorter studies, as can be seen in Figure \ref{figese2}  and Table \ref{tablelabe2} for log-normal errors.
Our simulations shed some light on the behaviour of our proposed estimators that make no assumptions on the error distribution. In the Supplementary material, one can see that while the method in C\&S is robust to some misspecifications, it is less so for others, while our method is more consistent in its bias, regardless of the error distribution. By comparing the log-normal cases for short studies ( Table \ref{tablelabe2}  here and Table 5 in the Supplementary material), we conclude that there may be reason to prefer our method on samples with fewer events and asymmetric errors.
In all tables (see also Supplementary material) we use the following abbreviations:

ENOES: Expected number of events.

$|rBias|$: Absolute relative bias: $\left|(O-E)/E \right| $ where $O$, $E$ stand for ``Observed'' and ``Expected'', respectively.

ESE: Empirical standard error.

ASE: Asymptotic standard error (when available).

\begin{table}[h]
\caption{ Comparison with C\&S, $n=50$, $V_{ij}$ in (\ref{equation1.5}), $ C_{\max}=125$ }%
\resizebox{\columnwidth}{!}{
\begin{tabular}{lcccccccc}
  {}& ENOES=3.9{} &   {} & Our method &  {} &  {}   & {} & $C\&S$ $F_0=$ Normal & {}\\
  $Estimator$ & Parameter &  $|rBias|$ & ESE & Bias & {} & $|rBias|$ & ESE & Bias\\
  Normal errors & $\gamma_0$  & 0.027 & 0.772 & -0.016 & {}  & 0.051 & 1.246 & -0.030\\
  {$\mu_{ij}(\theta):=(\ref{equation1.4})$} & $\gamma_1$ & 0.074 & 1.027 & 0.029& {}  & 0.095 & 1.289 & 0.038\\
  {$V_{ij}(\theta):=(\ref{equation1.5})$} & $\rho$ & 0.500 & 0.122 & 0.015 & {}  & 0.084 & 0.067 &0.003\\
  {} & $\sigma^2$  & 2.376& 1.471 & 7.843 & {} & 2.305 & 4.431 & 7.606\\
Exponential errors & $\gamma_0$  & 0.027 & 0.763 & -0.016 & {}  & 0.182 & 0.709& -0.109\\
  {$\mu_{ij}(\theta):=(\ref{equation1.4})$} & $\gamma_1$ &0.078 & 1.015 & 0.031 & {}  & 0.348 & 0.770& 0.139\\
  {$V_{ij}(\theta):=(\ref{equation1.5})$} & $\rho$ & 0.988 & 0.151 & 0.030 & {}  & 0.096 & 0.076 & -0.003\\
  {} & $\sigma^2$  & 2.328 & 2.482 & 7.681 & {} & 2.154 & 2.081 & 7.110\\
Uniform errors & $\gamma_0$  & 0.017 & 0.766 & -0.010 & {}  & 0.012 & 0.455 & 0.007\\
  {$\mu_{ij}(\theta):=(\ref{equation1.4})$} & $\gamma_1$ &0.076 & 1.033 & 0.030 & {}  & 0.022 & 0.694 & 0.009\\
  {$V_{ij}(\theta):=(\ref{equation1.5})$} & $\rho$ & 0.311 & 0.119 & 0.009 & {}  & 0.092& 0.059& -0.003\\
  {} & $\sigma^2$  & 2.381 & 1.078& 7.856 & {} & 2.357 & 0.812 & 7.778\\
Log-normal errors & $\gamma_0$  &  0.033 & 0.749 & 0.020& {} & 0.145 & 0.430 & -0.087\\ 
  {$\mu_{ij}(\theta)=(\ref{equation1.4})$} & $\gamma_1$& 0.054 & 1.011 & -0.022& {} & 0.291 & 0.658 & 0.116\\ 
  {$V_{ij}(\theta)=(\ref{equation1.5})$} & $\rho$ &  1.077 & 0.219 & 0.032 &{}& 0.077 & 0.064 & -0.002\\ 
  {} & $\sigma^2$  & 0.003 & 3.992 & 0.028 &{}& 0.074 & 2.665 & -0.814\\ 
\end{tabular}
\label{tablelabe2}
}
\end{table}
One  limitation of these simulations is that only one set of values of $\eta$ is used. Thus, we are not assessing behaviour under the null hypothesis or within a range of weak to strong effects. Additionally, having only one deterministic covariate for all subjects is  an assumption of lack of noise. However, modeling recurrent events with a ``full history'' as in  (\ref{equation1.4}) is quite complex even with a single covariate, as the number of entities to input into the formula grows linearly with the number of events. In our simulations, individuals with longer past gap times will tend to have the mean log-gap  longer as well. The effect size of this is harder to explain in simple terms to a data user, unlike, say, the effect size from a proportional hazard model. It is not surprising that the default setup of the \texttt{condGEE} package actually models independent gap times within individuals, even if as a model it is rather uninteresting.  A future project could attempt to apply our method to models with marginal conditional covariance dependent on the conditional mean.

\section{The conditional approach, comparisons and conclusion}\label{coditionalcompared}
\subsection{Our conditional approach}

In this section we show how for each subject $i$ in (\ref{eauation.1.18}) and (\ref{equation.1.20}), 
 the sum in $j$ can be expressed as the conditional expectation of the corresponding sum in (\ref{equation1.9}) or (\ref{equation1.10}), respectively, given a $\sigma$-field $\mathcal{O}_{\tau_i}$, which will be defined below.

For simplicity and since we work at the level of a subject, we omit writing the index $i$,  as well as the parameters $\theta$ and $\eta$. We present the results only for the estimating function (\ref{equation1.9}), as (\ref{equation1.10}) is treated in a similar manner. Recall that, with the obvious notation,
\begin{equation}
\label{section6.1}
g_1= \sum_{j=1}^{\infty} f_jZ_jI\{S_{j-1}< C\}=\sum_{j=1}^{\infty}(f_jZ_j I_{j}^{obs}+f_jZ_jI_{j}^{cen}) =\sum_{j=1}^{\tau}f_j Z_j
\end{equation}
where we used (\ref{equation1.9}), (\ref{sequation.2.8newf}) and (\ref{equation1.16}) . We will be dealing with $\displaystyle g_{1}^{[l]}$, defined in (\ref{equation2.11.newf}), $(l=1,\ldots)$. When $l\rightarrow\infty$, we have that, $a.s.$,
 \begin{equation}
 \label{section6.2}
 g_{1}^{[l]}=\sum_{j=1}^l f_{j}Z_{j}I\{S_{j-1}<C\}\rightarrow \sum_{j=1}^{\tau}f_jZ_jI\{S_{j-1}<C\}
 \end{equation}

To define $\mathcal{O}_{\tau}$(see also p.219 of \cite{Williams:1995}), we proceed by defining first a sequence of $\sigma$-fields $\mathcal{O}_n$, generated as follows:
\begin{equation}
\label{section6.3}
\mathcal{O}_n =\sigma( f_jZ_jI_{n}^{cen}, j\leq n-1, f_nI_{n}^{cen}, b_nI_{n}^{cen}, I_{n}^{cen}), \quad n\geq 1.
\end{equation}
Let $\displaystyle \mathcal{O}_{\tau}= \{ A\subseteq \Omega: A \cap\{\tau=n\}\in \mathcal{O}_n, n\geq 1\}$.
We note that $\mathcal{O}_{\tau}$ is a $\sigma$-field because $\mathcal{O}_n$ is, for each $n\geq 1$. Perhaps surprisingly, the building blocks of our `` observed '' $\sigma$- field $\mathcal{O}_{\tau}$ are the censoring events.  Although $I\{\tau=n\}= I_{n}^{cen}$ is $\mathcal{O}_n$-measurable, $\tau$ is not a stopping time with respect to $\{\mathcal{O}_n\}_{n\geq 1}$, which is not a filtration.
\begin{lemma}
\label{lemma6.1}
The functions $\displaystyle f_jZ_jI_{j}^{obs}, f_{j}I_{j}^{cen}, b_jI_{j}^{cen}, I_{j}^{cen}$  are $\mathcal{O}_\tau$- measurable, $(j=1,\ldots).$
\end{lemma}
\begin{proof}
We show first that $\displaystyle f_jZ_jI_{j}^{obs}$ is $\mathcal{O}_\tau$-measurable, for each $j\geq 1$.\\
$ \displaystyle \quad
\mbox{If}\; \; \;\;n\leq j,\quad    \{S_{n-1}<C < S_n\}\cap \{S_j\leq C\}=\phi\in \mathcal{O}_n.\nonumber $\\
$\displaystyle \quad
\mbox{If}\;\;\;\; n>j,\quad
\{S_{n-1}< C< S_n\} \subseteq \{S_j\leq C\},\;\; \mbox{so}\;\;  f_jZ_jI_{j}^{obs}I_{n}^{cen}= f_jZ_jI_{n}^{cen},$\\
 which is a generator of $\mathcal{O}_n$ in (\ref{section6.3}). To avoid repetition, we only show that $f_jI_{j}^{cen}$ is $\mathcal{O}_\tau$- measurable.\\
   $\displaystyle \quad \mbox{If}\;\;\;\; n\neq j, \quad \{S_{n-1}<C<S_n\}\cap\{S_{j-1}< C< S_n\}=\phi\in \mathcal{O}_n.\; \mbox{When} \;\; n=j,$ \\
    $ f_jI_{j}^{cen}I_{n}^{cen}=f_nI_{n}^{cen}$, a generator of $\mathcal{O}_n$.
\end{proof}

\begin{remark}
The $\sigma$-fields $\mathcal{O}_n$ are generated by some observed functions, and $\mathcal{O}_\tau$ tracks the evolution in time of a subject. We note that $I_{n}^{cen}$ is observed, even if the value of $S_n$ is not. We see at the time of censoring that $S_n$ will be larger than C. We therefore can observe $f_jZ_jI_{n}^{cen}$, as $f_j$ is observed at time $S_{j-1}$ and $Z_j$ at time $S_j , j\leq n-1$. On the other hand, $Z_nI_{n}^{cen}$ is not observed, because the unobserved value $S_n$ is needed in the calculation of $Z_n$.
\end{remark}

\begin{theorem}
Assume that Condition \ref{Condition2.5} holds.  Furthermore, assume that (B) holds:

\begin{equation}
(B)\quad\quad\quad \quad \quad\quad\qquad  E[Z_jI_{j}^{cen}|\mathcal{O}_\tau]= E[Z_jI_{j}^{cen}|I_{j}^{cen}] \quad\quad \quad\quad  (j=1,\ldots).\nonumber
\end{equation}
Then
\begin{equation}
\label{section6.5}
 E[g_1|\mathcal{O}_\tau]=\sum_{j=1}^{\tau}\bigg(f_jZ_jI_{j}^{obs}+ f_jI_{j}^{cen}\frac{E[Z_jI_{j}^{cen}]}{E[I_{j}^{cen}]}\bigg)
\end{equation}
In addition, $E[g_1]=0$, so $g_1$ in (\ref{section6.1}) is unbiased.
\end{theorem}
\begin{proof}
Recall Proposition \ref{Proposition1.1}, (\ref{equation2.11.newf}) and (\ref{equation1.14}). By (g), p.88 of \cite {Williams:1995}, \\$\displaystyle \lim_{l\rightarrow \infty}E[g_{1}^{[l]}|\mathcal{O}_\tau]=E[g_1|\mathcal{O}_\tau]$, and $ \displaystyle \lim_{l\rightarrow \infty}E[g_{1}^{[l]}]=E[g_1]$. From Lemma \ref{lemma6.1}, \\for $(l=1,\ldots,),$
\begin{equation}
\label{section6.4new}
E[g_{1}^{[l]}|\mathcal{O}_{\tau}]=\sum_{j=1}^l\bigg(f_jZ_jI_{j}^{obs}+f_jI_{j}^{cen}E[Z_jI_{j}^{cen}|\mathcal{O}_{\tau}]\bigg)
\end{equation}

By (B), each second term in (\ref{section6.4new}) is of the form $\displaystyle f_jI_{j}^{cen}E[Z_jI_{j}^{cen}]/E[I_{j}^{cen}]$. Thus, $ \displaystyle \lim_{l\rightarrow \infty}
E[g_{1}^{[l]}|\mathcal{O}_\tau]$ equals the right hand side of (\ref{section6.5}), as well as $E[g_1|\mathcal{O}_\tau]$,  which proves (\ref{section6.5}). The second assertion follows as in the proof of Proposition \ref{Proposition1.1}. Alternatively, we could apply (g) p.88 of \cite{Williams:1995} with the $\sigma$-field $\{\Omega,\phi\}$
\end{proof}
\begin{remark}
Condition $(B)$ is realistic in describing the impact of an important event on the immediate development of the process. It is a conditional independence assumption which states that, when censoring occurs, the length of the would-be gap time is independent of the history of the subject imbedded in $\mathcal{O}_\tau$(see 34.11 of \cite{Billingsley:1995}).
\end{remark}
\subsection{Comparison with similar results and conclusion}
In this section we compare our wok with the results in \cite{clement and Strawderman} and \cite{Murphy.etal:1995}, starting with the conditional estimating functions used. We begin by showing on $g_1$  in (\ref{section6.1}) how conditioning on an appropriate $\sigma$-field leads to (2.12) in \cite{clement and Strawderman}, and to $S_I(\beta)$ on p.1848 of \cite{Murphy.etal:1995}, where $\beta$ is the regression parameter.

We first define the conditional $\sigma$-field which appears on p. 454 of \cite{clement and Strawderman}. In our notation, $N+1=\tau$, and $\{W_{\tau}(\eta)> w(\eta)\}$ reduces to $\displaystyle \{S_{\tau}>C\}=\Omega$, by the definition of $\tau$. Then the conditioning $\sigma$-field in \cite{clement and Strawderman} is:
\begin{equation} H_{\tau}=\{ A \subseteq \Omega: A \cap \{\tau=n\}\in  \mathcal{O}_{n}^{[1]}\}, \quad\quad \mbox{with}\nonumber
\end{equation}
\begin{equation} \mathcal{O}_{n}^{[1]}=\sigma( F\cap \{\tau=n\}, f_n I_{n}^{cen}, b_nI_{n}^{cen}, F \in \mathcal{F}_{n-1}), \;\;\; (n=1,\ldots).\nonumber
\end{equation}
In $\displaystyle g_1= \sum_{j=1}^{\tau-1}f_jZ_j + f_{\tau}Z_{\tau}, I_{n}^{cen}(\sum_{j=1}^{n-1}f_jZ_j) $ is $\mathcal{O}_{n}^{[1]}$- measurable. Likewise, $f_{\tau}$ is $H_{\tau}$-measurable, since $\displaystyle I_{n}^{cen}f_{\tau}=I_{n}^{cen}f_n$, which is a generator of $\mathcal{O}_{n}^{[1]}$. Then
\begin{equation}
\label{equation6.6newf}
E[g_1|H_{\tau}] =\sum_{j=1}^{\tau-1} f_j Z_j + f_{\tau}E[ Z_\tau|H_{\tau}],
\end{equation}
which is essentially (2.8) of \cite{clement and Strawderman}. The unbiasedness in (\ref{equation6.6newf}) is stated in their Theorem 2.1, and proved in Supplementary material from \cite{clement and Strawderman}, under the additional assumption $(T1)$. For the estimating function in (2.9), \cite{clement and Strawderman} gives:

\begin{equation}
\label{equation6.7newf}
E[g_2|H_{\tau}] =\sum_{j=1}^{\tau-1} b_j( Z_{j}^2-\sigma^2) + b_{\tau}(E[Z_{\tau}^2|H_{\tau}]-\sigma^2)
\end{equation}

The conditioning $\sigma$-field $H_{\tau}$ in \cite{clement and Strawderman} is similar to our $\mathcal{O}_{\tau}$ in Section 6.1, but the approach to the treatment of the censored gap times is different. The terms of the sums in (\ref{equation6.6newf}--\ref{equation6.7newf}) represent the observed terms, while the last terms contain the first two conditional moments of $Z_{\tau}$, given the $\sigma$-field $H_{\tau}$ over the gap times at the time of censoring. With the notation
\begin{equation}
 W_j =\sigma^{-1}Z_j, \quad\quad\quad (j=1, \ldots)\nonumber
\end{equation}
(\ref{equation6.6newf}--\ref{equation6.7newf}) generate the terms corresponding to individuals in (2.12--2.13), or (2.14--2.15) of \cite{clement and Strawderman}. Specifically, let
\begin{equation}
\label{equation6.8new}
n^{-1}S_1=\sum_{j=1}^{\tau-1} f_jW_j+f_{\tau}K_1(w),
\end{equation}
\begin{equation}
\label{equation6.9new}
n^{-1}S_2=\sum_{j=1}^{\tau-1} b_j(W_{j}^2-1)+b_{\tau}(K_{2}(w)-1),
\end{equation}
where $\displaystyle w=\frac{C-S_{\tau-1}-\mu_{\tau}}{\sigma V_{\tau}}.$\\
In (\ref{equation6.8new}--\ref{equation6.9new}),  $\displaystyle K_{r}(w)=\int_{w}^{\infty}\frac{u^r}{1-F_0(w-)} dF_0(u),\; r=1,2$, where $F_0$ specifies the complete conditional distribution of $W_{\tau}$, given $H_{\tau}$, which is assumed to be known.

We now discuss conditioning in \cite{Murphy.etal:1995}. The estimating function $S_I(\beta)$ on p.1848 of \cite{Murphy.etal:1995} is our estimating function (\ref{eauation.1.18}), once we add in \cite{Murphy.etal:1995} the implicit restriction
$\displaystyle C_i-\sum_{l=1}^{j-1} Y_{il}>0 \; (i=1,\ldots).$ We propose that the conditioning field denoted ``obs" in \cite{Murphy.etal:1995}, but not spelt out there, be
our $\mathcal{O}_{\tau}$ in Section 6.1. To obtain the simple form of $S_I(\beta)$ we required assumption (B), which, with Condition \ref{Condition2.5} prove the unbiasedness of (\ref{eauation.1.18}) and $S_I(\beta)$. We note that neither condition(B) nor the unbiasedness of $S_I(\beta)$ is discussed in \cite{Murphy.etal:1995}.

Our Proposition \ref{proposition1.2}  proves this result also when (B) holds, since (B) implies both $(B_{f(\theta)})$ and $(B_{b(\eta)})$. Indeed, from (\ref{equation1.22}) and by Lemma \ref{lemma6.1}, we have :

\begin{equation}
\sigma(I_{ij}^{cen}, \;(j=1,\ldots))\subset \mathcal{O}(f_{i}(\theta))\subset \mathcal{O}_{\tau_i}\quad ( i=1, \ldots).\nonumber
\end{equation}
Then
\begin{eqnarray}
I_{ij}^{cen}E_{\theta}[Z_{ij}(\theta)|\mathcal{O}(f_{i}(\theta))]&=&I_{ij}^{cen}E_{\theta}[E_{\theta}[Z_{ij}(\theta)| \mathcal{O}_{\tau_i}]| \mathcal{O}(f_{i}(\theta))]\nonumber\\
&=&I_{ij}^{cen}E_{\theta}[E_{\theta}[Z_{ij}(\theta)| \sigma(I_{ij}^{cen}, j=1, \ldots)]| \mathcal{O}(f_{i}(\theta))]\nonumber\\
&=&I_{ij}^{cen}E_{\theta}[Z_{ij}(\theta)| \sigma(I_{ij}^{cen}, j=1, \ldots)],\nonumber
\end{eqnarray}
where we used (B) for the second equality. This proves $(B_{f(\theta)})$, and $(B_{b(\eta)})$ can be proved similarity.

Ultimately, it is the conditioning $\sigma$-fields of the partially observed gap times that matter, and ours  are identical to the ``obs"  in \cite{Murphy.etal:1995},
while $H_{\tau} $ in \cite{clement and Strawderman} is larger. Our smaller conditioning $\sigma$-fields are obtained by imposing condition (B).

The next step in the analysis with censored  data is replacing the conditioned terms in (\ref{equation6.6newf}--\ref{equation6.7newf}) with actual data, $i.e.$, imputing appropriate values in the estimating functions.
In \cite{clement and Strawderman} and \cite{Murphy.etal:1995}, the authors generate  data from a known distribution $F_0$ to replace the conditioned terms in (\ref{equation6.6newf}--\ref{equation6.7newf}), for all subjects. Thus, they use a parametric method of imputation.

The  method used in \cite{Murphy.etal:1995} to construct estimators of the regression and overdispersion parameters  is described succinctly in Section 2.4 of \cite{clement and Strawderman}. 
 The results from $n$ individuals are combined in  functions of type (\ref{eauation.1.18}), denoted $S_M(\theta)$ on p.457 of \cite{clement and Strawderman}. The estimating equation $S_M(\theta)=0$  is solved in $\theta$, which is then used to update $\hat{\sigma}$, as in (2.20) of \cite{clement and Strawderman}. A full solution $(\hat{\theta},\hat{\sigma})$ is now available, and the procedure continues with the next iteration, until the values of these estimators stabilize. A formal proof of convergence is not given.

We now compare our mathematical techniques with those in \cite{clement and Strawderman}, \cite{Murphy.etal:1995}, \cite{Balan.etal:2010} and  \cite{Xie and Yang:2003}. Technically, the most difficult part is proving the existence and consistency of the estimators, which are defined implicitly as roots of estimating equations. Conditions  for these results to hold are not given in \cite{Murphy.etal:1995}. An  appropriate comparison with our work can be made with the conditions in the Appendix of \cite{clement and Strawderman}. They base their theoretical results on \cite{Yuan and Jennrich: 1998}, while we rely on  the more recent results in \cite{Xie and Yang:2003}, generalized in \cite{Balan.etal:2005} and \cite{Balan.etal:2010}. Some of the conditions  in \cite{Yuan and Jennrich: 1998} are difficult to prove, $e.g.$, conditions (A2)-(A3) listed in \cite{clement and Strawderman}.  Our conditions in  Theorem \ref{theorem2} rely on analytical properties of the derivatives $\mathcal{D}_{n}(\eta)$ of our estimating functions, without appealing to the existence of a limiting process and any type of convergence to it. 
 Our conditions, $e.g.$ S(ii), can be expressed in terms of simpler random variables and deterministic functions that make-up our estimating functions(see Appendix \ref{app}).  We base our technical approach on \cite{Xie and Yang:2003},  \cite{Balan.etal:2005} and \cite{Balan.etal:2010}. However, we are dealing here with a more general model. In our case the conditional variance need not be analytically linked with the conditional mean.
 The function representing this mean may be nonlinear in the components of the parameters ( see Example \ref{example1.4} ) .

We prove all our results on the asymptotic normality of our estimators in Appendix \ref{appb.2}. Our Theorem  \ref{theroem5} gives the asymptotic normal distribution of  $\hat{\sigma}_{n}^2$ , taking into account the asymptotic distribution of $\hat{\theta}_n$. We followed
ideas from Section 3.5 of \cite{Yi.et al: 2012} and provided the technical details, as outlined in \cite{Park.J:2018}.

We conclude that our proposed method is a valid alternative to the methods  proposed in \cite{clement and Strawderman} and \cite{Murphy.etal:1995}. It has the advantage of employing a nonparametric imputation method and it relies on proven mathematical results. The numerical results for a small sample size and a short study are good.

\appendix

\section{Results for the derivative}\label{app}
We present below some conditions for the derivative of $\hat{g}_{n,1}^{obs}$ to satisfy $S(ii)$ of Theorem \ref{theorem2} for Example \ref{example1.5}. Technical results of this nature appeared often in the literature, starting with the work in \cite{Xie and Yang:2003}. A more recent reference is to lemmas 4.7-4.9 in \cite{Balan.etal:2010}. What is required is a mixture of analytical and stochastic conditions on the random variables that make up the terms of the estimating functions.  Let $X_i$ denote the $\tau_i\times q$ matrix of covariates, with $jh$ entry $\displaystyle (X_i)_{jh}=x_{ij,h},(j=1,2\ldots,\tau_i;h=1,\ldots,q)$ and $(i=1,\ldots),$ where $q$ is the maximum number of covariates over all individuals, assumed to be finite. Note that $\displaystyle \sum_{j=1}^{\tau_i}\| x_{ij}\|^2=tr(X_{i}^TX_i)$.

The stochastic condition needed here is:
\begin{equation}
\label{A1}
E [tr(X_{1}^TX_1)]<\infty
\end{equation}
This assumption ensures the $a.s.$ convergence of $\displaystyle n^{-1}\sum_{i=1}^n tr(X_{i}^TX_i)$ to the left hand side of $(\ref{A1})$, by the SLLN.

For ease of notation, we use $\mu_{ij}^{[s]}(\theta)$ for the $s^{th}$ derivative of $\mu: R\rightarrow R,$ assumed to be continuous and evaluated at $c_{ij}^T(\theta)x_{ij}$, and $B_r$ for $B_{r}(\theta_0), r> 0$. Let $h_{ij,k}(\theta)\in R$ be a family of functions indexed by $k\in R^{l},(i,j,=1,\ldots).$ We define
\begin{equation}
b_{r,n}(h)=\sup_{\theta',\theta\in B_r}\max_{1\leq i\leq n,1\leq j\leq \tau_i}V_{ij}^{-1}(\theta')\|h_{ij}(\theta)\|_l,\nonumber
 \end{equation}
 where $\displaystyle \|h_{ij}(\theta)\|_l=\max_{1\leq k\leq l}\mid h_{ij,k}(\theta)\mid$, which is equivalent to the Euclidean norm in $R^{l}$. For instance, with $\displaystyle h=\partial V/\partial\theta,\;\|h_{ij}(\theta)\|_p=\max_{1\leq a\leq p}\mid \partial V_{ij}(\theta)/\partial \theta_a\mid$. When $\displaystyle h=V,\;b_{r,n}(V)= \sup_{\theta',\theta\in B_r}\max_{1\leq i\leq n,1\leq j\leq \tau_i}V_{ij}^{-1}(\theta')V_{ij}(\theta)$.
For $c_{ij}(\theta)\in R^q$, we do not use a normalizer and define instead :
\begin{equation}
c_{n}^{[3]}(r)=\sup_{\theta\in B_{r}(\theta)}\max_{1\leq i\leq n,1\leq j\leq\tau_i}\max_{a,b,d\leq p, h\leq q}\bigg | \frac{\partial ^3c_{ij,h}(\theta)}{\partial \theta_a\partial \theta_b\partial \theta_d}\bigg |,\nonumber
\end{equation}
where we assumed that all partial derivatives are continuous. We define similarly $c_{n}^{[t]}(r), (t=1,2)$.

We now state the equiboundedness condition $(A2)$:
\begin{eqnarray}
\label{A2}
\lim_{r\rightarrow 0}\limsup_{n\rightarrow \infty}\max\{b_{r,n}(h), c_{n}^{[t]}(r)\}<\infty,
\end{eqnarray}
which holds for
\begin{equation}
h=\mu, \mu^{[s]}, V,\partial V/\partial \theta, (s, t=1,2,3).\nonumber
\end{equation}
We define  moduli of  equicontinuity  at $\theta_0\in R^P$. With $h_{ij,k}(\theta)$ as before, let:
\begin{equation}
\delta_{r,n}(h)=\sup_{\theta',\theta\in B_r}\max_{1\leq i\leq n,1\leq j\leq \tau_i}V_{ij}^{-1}(\theta')\|h_{ij}(\theta)-h_{ij}(\theta_0)\|_l.\nonumber
\end{equation}
For instance, $\displaystyle \delta_{r,n}(V)=\sup_{\theta',\theta\in B_r}\max_{1\leq i\leq n,1\leq j\leq \tau_i}V_{ij}^{-1}(\theta')|V_{ij}(\theta)-V_{ij}(\theta_0)|.$

As before, we define separately the moduli:
\begin{equation}
\delta(c_{n}^{[1]}(r))=\sup_{\theta\in B_r}\max_{1\leq i\leq n, 1\leq j\leq\tau_i}\bigg\|\frac{\partial c_{ij}^{T}(\theta)}{\partial \theta}-
\frac{\partial c_{ij}^{T}(\theta_0)}{\partial \theta}\bigg\|_{pq},\nonumber
\end{equation}
with $\delta(c_{n}^{[2]}(r)) $ defined in a similar manner.  The generic equicontinuity condition is:
 \begin{eqnarray}
 \label{A3}
\lim_{r\rightarrow 0}\limsup_{n\rightarrow \infty}\max\{\delta_{r,n}(h), \delta (c_{n}^{[t]}(r))\}=0,\nonumber\\
 \mbox{ with}\;\;
h=\mu, \mu^{[s]}, V,\partial V/\partial \theta, \;(s, t=1,2).
\end{eqnarray}
 If (\ref{A2}) holds, a less restricted form of (\ref{A3}) can be used, as equiboundedness of derivatives at some level implies equicontinuity at one level below. We now state:
\begin{proposition}
Assume that $(T2)$ holds, along with conditions (\ref{A1})--(\ref{A3}). Then condition $S(ii)$ of Theorem \ref{theorem2} holds for the derivative of $\hat{g}_{n,1}^{obs}(\theta)$.
\end{proposition}
This result covers the content of Proposition 3.2 in \cite{Liu.etal:2018}.  
We streamlined in (\ref{A2})--(\ref{A3}) the notation and the conditions found in \cite{Liu.etal:2018}. The proof of this result can be found in lemmas 3.2.1--3.2.5 and 3.2.9--3.2.10 for the nonimputed and the imputed part of the derivatives, respectively, all in \cite{Liu.etal:2018}.

\begin{example}
\label{appendixexample1}
We illustrate conditions (\ref{A2})--(\ref{A3}) on Example \ref{example1.4} with formula (\ref{equation1.5}). We assume that $(T2)$ holds. Without loss of generality, we restrict $\rho$ to the interval $(0,1)$. We denote by $\displaystyle \dot{f}_j,\ddot{f}_j,\dddot{f}_j$ the first three derivatives of $ f_j$ in (\ref{equation1.3}). We first examine $\displaystyle V_{ij}(\rho)$, and show that there exist $r_0>0$ and constants $0<C_1<C_2$, such  that, for $0<r\leq r_0$,
\begin{equation}
\label{A4}
C_1\leq \sup_{\rho\in B_{r}(\rho_{0})}\max_{1\leq i\leq n; 1\leq j\leq\tau_i}V_{ij}^{-1}(\rho)\leq C_2.
\end{equation}
 Let $0<\rho_{0}<1$ and take $0<r_{0}<\min\{\rho_{0},1-\rho_{0}\}$. As $\dot{f}_{j}(\rho)>0$,  $f_{j}(\rho)$ in (\ref{equation1.3}) is increasing in $\rho$ and decreasing in $j$. For $\rho\in B_r(\rho_{0}),r\leq r_0,$
\begin{equation}
1+f_j(\rho_0-r_0)\leq V_{ij}^2(\rho)\leq 1+f_{j}(\rho_0+r_0)\;\;(j=1,\ldots,\tau_i, i=1,\ldots,n).\nonumber
\end{equation}
Then
\begin{equation}
 C_{2}^{-2}=1\leq 1+f_{j}(\rho_{0}-r_{0})\leq V_{ij}^2(\rho)\leq 1+f_{1}(\rho_{0}+r_0)=C_{1}^{-2}.\nonumber
\end{equation}

The first inequality above holds because, in (\ref{equation1.3}),
$\displaystyle (\rho_{0}-r_{0})(j-2)+1>0. $

We obtain (\ref{A4}) by taking reciprocals. So (\ref{A2}) holds for $V$, and, for some $C>0$,
\begin{equation}
\delta_{r,n}(V) \leq \sup_{\rho \in B_r(\rho_0)}\max_{1\leq i\leq n,1\leq j\leq\tau_i,} C\mid V_{ij}(\rho) - V_{ij}(\rho_0)\mid \nonumber
\end{equation}
We now show (\ref{A3}).
Since $\displaystyle f_{j}(\rho_0+r_0)\rightarrow 0$ as $\displaystyle j\rightarrow \infty$, for $\displaystyle 0<\varepsilon < \rho_0$, we can find a first integer $\displaystyle j_0=j(r_0,\varepsilon)\geq 1/\varepsilon-1/(\rho_0+r_0)+2$, with $\displaystyle f_j(\rho_0+r_0)\leq \varepsilon$ $\displaystyle(j=j_0,\ldots)$. If $\displaystyle r\leq r_0$, then $\displaystyle j(r,\varepsilon)\leq j_0$. Since $f_j$ is continuous  at $\rho_0$ for each $j$, let $\displaystyle r(\varepsilon, j_0)>0$ be such that, if
$\displaystyle \rho\in B_{r(\varepsilon,j_0)}(\rho_0)$, $\displaystyle \max_{1\leq j\leq j_0}\mid V_{j}(\rho)-V_{j}(\rho_0)\mid\leq \varepsilon$. Since $f_j$ is increasing in $\rho$ and from the definition of $j_0$,
$\displaystyle \sup_{j\geq 1}\mid V_{j}(\rho)-V_{j}(\rho_0)\mid \leq 2 \varepsilon$ which, by (\ref{A4}), proves (\ref{A3}) for V. As (\ref{A2}) holds for V, we ignore it as normalizer for other functions which appear in (\ref{A2}) and (\ref{A3}).

Through a similar calculation we can show that (\ref{A2}--\ref{A3}) hold for $\dot{V}$.

We now examine conditions for $\mu$ to satisfy (\ref{A2}--\ref{A3}). Since $\mu$ is the identity function here, its derivatives satisfy (\ref{A2}--\ref{A3}). When $(T2)$ holds, not only is $j$ or $\displaystyle \sup_{i}\tau_i $ bounded for all individuals, but so is the number of components of $\displaystyle c_j(\theta)$, which depends on $j, j=1,\ldots,m$. Therefore $\displaystyle c_j(\theta) \in R^q$, for some $q\geq 1$, for all $j$ and $\theta \in R^3$. From examples \ref{example1.4}  and  \ref{example1.5}, all components $\displaystyle c_{j,h}(\theta)$ are continuous at $\theta_0$, so for this finite family, we have (\ref{A2}--\ref{A3}) for $c$. Now

 \begin{equation}\mid \mu_{ij}(\theta)\mid= \mid c_{ij}^T(\theta)x_{ij}\mid\leq \parallel c_{ij}(\theta)\parallel_q\parallel x_{ij}\parallel_q\leq\max_{j\leq \tau_i,h\leq q}\mid c_{j,h}(\theta)\mid \parallel x_{ij}\parallel_q,\nonumber
 \end{equation}

 \begin{equation}
 b_{r, n}(\mu)\leq b_{r}(c)\max_{1\leq i\leq n, j\leq \tau_i}\parallel x_{ij}\parallel_q.\nonumber
 \end{equation}
 It means that $\mu$ satisfies (\ref{A2}), if $x$ satisfies (\ref{A2}), $i.e.$, if :
\begin{equation}
\label{A5}
\limsup_{n\rightarrow \infty}\max_{1\leq i\leq n, j\leq\tau_i}\parallel x_{ij}\parallel_q< \infty.
\end{equation}

Similarly, $\displaystyle \delta_{r,n}(\mu)\leq \delta_r(c)b_{r,n}(x)$ and since  $\displaystyle \lim_{r\rightarrow 0}\delta_r(c)=0$, $\mu$ satisfies (\ref{A3}), if (\ref{A5}) holds. We note that (\ref{A5}) holds in this example, because we can conceive of a universal upper bound for the BMI's for all individuals. For the function $c$, it suffices to show that the first three partial derivatives satisfy (\ref{A2}), as, by the mean value theorem, this implies (\ref{A3}) for the first two partial derivatives, calculated from Example \ref{example.1.5new} . Since $\displaystyle \theta^T=(\gamma_0,\gamma_1, \rho), \theta \in B_{r}(\theta_0)$ and $j$ is bounded,

 \begin{equation}
  \parallel\partial c_{j}(\theta)/\partial \theta\parallel_3\leq C\max_{1\leq j\leq m}\max \{f_{j}(\rho),\dot{f}_{j}(\rho), r\}, \nonumber
  \end{equation}
 $ C> 0$  a constant.

Now (\ref{A2}) follows from the equicontinuity and equiboundedness of the functions on the right hand side. Likewise, the finite set of second or third order derivatives is  bounded by $\displaystyle C \max\{ \dot{f}_j(\rho),\ddot{f}_{j}(\rho)\}, C\max\{\ddot{f}_{j}(\rho), \dddot{f}_{j}(\rho)\}$, respectively.
\end{example}

\begin{remark}
\label{remarkA.1}
Assume that (\ref{A4}) holds. From the definitions of $\delta_{r,n}(h)$ and $b_{r,n}(h)$ in Lemma \ref{lemma2},  it suffices to show that (\ref{A2}--\ref{A3}) hold for $c_{n}^{[1]}(r)$.
\end{remark}

\section{Proofs of asymptotic normality}
\label{appb.2}
\subsection{Proof of Theorem \ref{equ.th.4.1.1}}

 We write $\displaystyle u_i=u_{i,1}$, with $\displaystyle u_{i}^{T}=(u_{ik})_{k=1,2,\cdots p}$, $i\geq 1$. Since $E_{\theta_0}(u_{1})=0$ by Proposition \ref{proposition1.2}, the proof of Theorem \ref{equ.th.4.1.1} follows from Theorem 29.5 of \cite{Billingsley:1995}, once we show that
\begin{equation}
\max_{1\leq k\leq p}E_{\theta_0}[u_{1k}^2]<\infty.\nonumber
\end{equation}
This will follow if all entries of $\Sigma$ have finite expectations. For a fixed $i$, and $1\leq k,l\leq p,$  to calculate the expectation of
$\displaystyle u_{ik}u_{il}$, we use
\begin{eqnarray}
&&I_{ij}^{obs}I_{ij'}^{obs}=I_{i\max\{j,j'\}}^{obs},\quad \quad   I_{ij}^{obs}I_{ij'}^{cen}=I_{ij'}^{cen},\;\;\;\; \mbox{if}\;\;\;\; j'>j,\nonumber\\
&&I_{ij}^{cen}I_{ij'}^{cen}=I_{ij}^{cen},\;\; \quad \mbox{if}\;\; j=j',\;\; \mbox{and}\; \;0 \;\;\mbox{otherwise},\nonumber
\end{eqnarray}
then take expectations and obtain (\ref{eqn.n.4.1.4.1.26}) as the $(kl)$ entry of $\Sigma$, which exist by (\ref{equ.con.4.1.1.f}).

\subsection{Proof of Proposition \ref{Pro.4.1.2}}

Since
\begin{equation}
n^{-1/2}\hat{g}_{n,1}^{obs}(\theta_{0})=n^{-1/2}g_{n,1}^{obs}(\theta_{0})+n^{-1/2}(\hat{g}_{n,1}^{obs}(\theta_{0})-g_{n,1}^{obs}(\theta_{0})),\nonumber
\end{equation}
by Theorem \ref{equ.th.4.1.1} and Slutsky's theorem, it suffices to prove
\begin{equation}
n^{-1/2}(\hat{g}_{n,1}^{obs}(\theta_{0})-g_{n,1}^{obs}(\theta_{0}))\rightarrow 0 \;\mbox{in probability.}\nonumber
\end{equation}
Now, $ \displaystyle n^{-1/2}(g_{n,1}^{obs}(\theta_{0})-\hat{g}_{n,1}^{obs}(\theta_{0}))=\sum_{j=1}^{m}\bigg (n^{-1/2}\sum_{i=1}^nf_{ij}(\theta_0)I_{ij}^{cen}\bigg)\delta_{n,j}^{cen},
$

where
$\displaystyle
\delta_{n,j}^{cen}=-E_{\theta_0}[Z_{1j}(\theta_0)I_{1j}^{obs}]/
E_{\theta_0}[I_{1j}^{cen}]+\sum_{k=1}^{n}Z_{kj}(\theta_0)I_{kj}^{obs}/
\sum_{k=1}^{n}I_{kj}^{cen}.$
With condition (\ref{equ.con.4.1.1.f}), by Theorem 29.5 of \cite{Billingsley:1995}, for $1\leq j\leq m,$
\begin{equation}
n^{-1/2}\sum_{i=1}^nf_{ij}(\theta_0)I_{ij}^{cen}\; \rightarrow\; N(E_{j}^{cen}(\theta_0),\Sigma_{j}^{cen}(\theta_0)),\mbox{ in distribution, }\nonumber
\end{equation}
 where the $k^{th}$ component of $ \displaystyle E_{j}^{cen}(\theta_0)$ is $ \displaystyle  E_{\theta_0}[f_{1j,k}(\theta_0)I_{1j}^{cen}]$ and the $kl$-entry of $\displaystyle \Sigma_{j}^{cen}(\theta_0)$ is, for $1\leq k,l\leq p,$
\begin{equation}
E_{\theta_0}[ (f_{1j,k}(\theta_0)I_{1j}^{cen}-E_{j,k}^{cen}(\theta_0))(f_{1j,l}(\theta_0)I_{1j}^{cen}-E_{j,l}^{cen}(\theta_0))].\nonumber
\end{equation}
For each $j$, $\delta_{n,j}^{cen}\rightarrow 0$ $a.s.$, when $n\rightarrow\infty$, by the SLLN.  Therefore, by Theorem 7.7.1 of \cite{Ash and Doleans:2000}, each of the terms in the sum above converges to zero in distribution, hence in probability.

\subsection{Proof of Theorem \ref{theorem.clt.4.13}}

We use the mean value theorem to write, with $\hat{\theta}_{n}$ given by Theorem \ref{theorem2},
\begin{equation}
\hat{g}_{n,1}^{obs}(\hat\theta_{n})-\hat{g}_{n,1}^{obs}(\theta_{0})=-\mathcal{D}_{n,1}(\bar{\theta}_{n})(\hat\theta_{n}-\theta_{0}),\;\mbox{where}\;||\bar{\theta}_{n}-\theta_0||\leq ||\hat{\theta}_{n}-\theta_0||.\nonumber
\end{equation}
Then
\begin{eqnarray}
\hat{g}_{n,1}^{obs}(\theta_{0})&=&\mathcal{D}_{n,1}(\bar{\theta}_{n})(\hat\theta_{n}-\theta_{0})\nonumber\\
&=&n\bigg[n^{-1}(\mathcal{D}_{n,1}(\bar{\theta}_{n})-\mathcal{D}_{n,1}(\theta_{0}))+n^{-1}\mathcal{D}_{n,1}(\theta_{0})\bigg](\hat\theta_{n}-\theta_{0})\nonumber\\
&=&n\bigg(o_P(1)+n^{-1}\mathcal{D}_{n,1}(\theta_{0})-\mathcal{D}_{1}(\theta_{0})+\mathcal{D}_{1}(\theta_{0})\bigg)(\hat\theta_{n}-\theta_{0})\nonumber
\end{eqnarray}

where we used condition $S(ii)$ and the conclusion of Theorem \ref{theorem2}. From the hypotheses,
\begin{equation}
\hat{g}_{n,1}^{obs}(\theta_{0})=n(o_P(1)+\mathcal{D}_{1}(\theta_0))(\hat\theta_{n}-\theta_{0}),\nonumber
\end{equation}
so
\begin{equation}
 n^{1/2}(\hat\theta_{n}-\theta_{0})=(o_P(1)+\mathcal{D}_{1}(\theta_0))^{-1}n^{-1/2}\hat{g}_{n,1}^{obs}(\theta_{0}).\nonumber
\end{equation}
By Proposition \ref{Pro.4.1.2},

 \begin{equation}
 n^{-1/2}\hat{g}_{n,1}^{obs}(\theta_{0})\rightarrow N(0, \Sigma). \nonumber
 \end{equation}
  By Theorem 3.2.1 of \cite{TYL:1999},

\begin{equation}
n^{1/2}(\hat\theta_{n}-\theta_{0})\rightarrow N (0, \mathcal{D}_{1}^{-1}(\theta_0)\Sigma( \mathcal{D}_{1}^{-1}(\theta_0))^T).\nonumber
\end{equation}

\subsection{Proof of Theorem \ref{theroem5}}

 By the mean value theorem and Theorem \ref{theorem2}, we obtain, for a large n,

\begin{eqnarray}\label{normaility.2.1}
&&-\hat{g}_{n,2}^{obs}(\theta_0,\sigma_{0}^2)=
\hat{g}_{n,2}^{obs}(\hat{\theta}_n,\hat{\sigma}_{n}^2)-\hat{g}_{n,2}^{obs}(\theta_0,\hat{\sigma}_{n}^2)
+\hat{g}_{n,2}^{obs}(\theta_0,\hat{\sigma}_{n}^2)-\hat{g}_{n,2}^{obs}(\theta_0,\sigma_{0}^2)\nonumber\\
&&=\frac{\partial \hat{g}_{n,2}^{obs}(\bar{\theta}_{n},\hat{\sigma}_{n}^2)}{\partial \theta^T}(\hat{\theta}_n-\theta_0)+\frac{\partial \hat{g}_{n,2}^{obs}(\theta_0,\bar{\sigma}_{n}^2)}{\partial\sigma^2}(\hat{\sigma}_{n}^2-\sigma_{0}^2),\nonumber
\end{eqnarray}

where $\displaystyle \parallel\bar{\theta}_{n}-\theta_0\parallel\leq\parallel\hat{\theta}_n-\theta_0\parallel$,
$\displaystyle \parallel\bar{\sigma}_{n}^2-\sigma_{0}^2\parallel\leq\parallel\hat{\sigma}_{n}^2-\sigma_{0}^2\parallel$.\\\\
Using (\ref{equ.new.final.4.2.15}), $S(ii)$ of Theorem \ref{theorem2} and the consistency of $(\hat{\theta}_n,\hat{\sigma}_{n}^2),$
\begin{eqnarray}
-\hat{g}_{n,2}^{obs}(\theta_0, \sigma_{0}^2)
&=&n\bigg (o_{P}(1)+\frac{1}{n}\frac{\partial \hat{g}_{n,2}^{obs}(\theta_0,\sigma_{0}^2)}{\partial \theta^T}\bigg)(\hat{\theta}_{n}-\theta_0)\nonumber\\
&+&n\bigg (o_{P}(1)+\frac{1}{n}\frac{\partial \hat{g}_{n,2}^{obs}(\theta_0,\sigma_{0}^2)}{\partial \sigma^2}\bigg)(\hat{\sigma}_{n}^{2}-\sigma_{0}^2)\nonumber\\
&=&n \bigg(o_P(1)-\mathcal{D}_{3}(\eta_0)\bigg)(\hat{\theta}_{n}-\theta_0)+n\bigg(o_P(1)-\mathcal{D}_{2}(\eta_0)\bigg)(\hat{\sigma}_{n}^{2}-\sigma_{0}^2)\nonumber
\end{eqnarray}
In the second equality  we used (\ref{equ.new.final.4.2.16}). Therefore,
\begin{eqnarray}
&&\bigg(\mathcal{D}_{2}(\eta_0)-o_{P}(1)\bigg)n(\hat{\sigma}_{n}^{2}-\sigma_{0}^2)\nonumber\\
&=&\hat{g}_{n,2}^{obs}(\theta_0, \sigma_{0}^2)
-\bigg(\mathcal{D}_{3}(\eta_0)-o_P(1)\bigg)n^{1/2}n^{1/2}(\hat{\theta}_{n}-\theta_0)\nonumber\\
&=&\hat{g}_{n,2}^{obs}(\theta_0, \sigma_{0}^2)-\bigg(\mathcal{D}_{3}(\eta_0)-o_P(1)\bigg)n^{1/2}\bigg(\mathcal{D}_{1}(\theta_0)+o_P(1)\bigg)^{-1}n^{-1/2}\hat{g}_{n,1}^{obs}(\theta_0)\nonumber\\
&=&\hat{g}_{n,2}^{obs}(\theta_0, \sigma_{0}^2)-\bigg(\mathcal{D}_{3}(\eta_0)-o_P(1)\bigg)\bigg(\mathcal{D}_{1}(\theta_0)+o_P(1)\bigg)^{-1}\hat{g}_{n,1}^{obs}(\theta_0),\nonumber
\end{eqnarray}
where we used  the proof of Theorem \ref{theorem.clt.4.13} in the second equality.
With the notation  above,
\begin{eqnarray}
&&\bigg(\mathcal{D}_{2}(\eta_0)-o_{P}(1)\bigg)n^{1/2}(\hat{\sigma}_{n}^{2}-\sigma_{0}^2)\nonumber\\
&=&n^{-1/2}\sum_{i=1}^{n}\bigg(\hat{u}_{i,2}(\eta_0)
-\mathcal{D}_{3}(\eta_0)\mathcal{D}_{1}^{-1}(\theta_0)\hat{u}_{i,1}(\theta_0)\bigg)+o_{P}(1).\nonumber
\end{eqnarray}
Recall that, by Theorem \ref{theorem.clt.4.13}, $n^{-1/2}\hat{g}_{n,1}^{obs}(\theta_0)$ is $O_{P}(1)$. As in the proof of Proposition \ref{Pro.4.1.2}, we can replace $\hat{u}_{i,1}(\theta_0) $ and $\hat{u}_{i,2}(\eta_0)$ above by $u_{i,1}$ and $u_{i,2}$, respectively. Now
\begin{equation}
n^{1/2}(\hat{\sigma}_{n}^{2}-\sigma_{0}^2)=\mathcal{D}_{2}^{-1}(\eta_0)n^{-1/2}\sum_{i=1}^{n}Q_i(\theta_0,\sigma_{0}^{2})+o_{P}(1).\nonumber
\end{equation}
Since $\displaystyle\{Q_i(\theta_0,\sigma_{0}^{2})\}_{i\geq 1}$ are $i.i.d.$ random variables with mean 0 and variance $\Omega$, by the central limit theorem,
\begin{equation}
n^{-1/2}\sum_{i=1}^{n}Q_i(\theta_0,\sigma_{0}^{2})\rightarrow  N(0,\Omega). \nonumber
\end{equation}

\section{Acknowledgements}
This research was in part supported by Natural Sciences and Engineering Research Council of Canada, notably through grant OGP009068 to M. Alvo, a grant awarded to P-J. Bergeron by the Canadian Institute for Health Research and a University of Ottawa Admission Scholarship awarded to H.Y. Liu. The authors would also like to thank Suzana Diaconescu for her contribution to examples \ref{example1.4}, \ref{example.1.5new} and \ref{appendixexample1} in this article.

\begin{supplement}\label{supplement}
\stitle{}
The model we are using is

\begin{eqnarray}
\label{eqn.7.1}
\mu_{ij}(\theta)&=&28+\gamma_0+\gamma_1\overline{BMI}_{ij}\nonumber\\
&+&\frac{\rho}{\rho(j-1)+1-\rho}
\bigg[\sum_{l=1}^{j-1}{Y_{il}}-\sum_{l=1}^{j-1}(28+\gamma_0+\gamma_1\overline{BMI}_{il})\bigg],
\end{eqnarray}

\begin{eqnarray}
\label{eqn.7.2}
\quad \quad\quad\quad\quad V_{ij}(\theta)=\left|1+\frac{\rho}{\rho(j-1)+1-\rho}\right|^{1/2},
\end{eqnarray}
To make things more compatible with modern packages in R, particularly for the \texttt{survival} package, the `` tidy data " format  of \cite{WickhamH:2014} is our format of choice.

\begin{table}[!htbp]
\centering
\caption{Simulated tidy data format }\label{parset}
\resizebox{\columnwidth}{!}{
\begin{tabular}{lcccccc}
 Subject ID & Start Time & End Time & Gap Time & Event Indicator & Event Number & $\overline{BMI}$\\
1 & 0 & $S_{1,1}$ & $Y_{1,1}$ & 1 & 1 & $\overline{BMI}_{1,1}$\\
1 & $S_{1,1}$ & $S_{1,2}$ & $Y_{1,2}$ & 1 & 2 &  $\overline{BMI}_{1,2}$\\
$\vdots$ & $\vdots$ & $\vdots$ & $\vdots$ & $\vdots$ & $\vdots$ & $\vdots$\\
1 & $S_{1,m_1-1}$ & $C_1$ & $Y_{1,m_1}$ & 0 & $m_1$ & $\overline{BMI}_{1,m_1}$\\
2 & 0 & $S_{2,1}$ & $Y_{2,1}$ & 1  & 1 & $\overline{BMI}_{2,1}$\\
$\vdots$ & $\vdots$ & $\vdots$ & $\vdots$ & $\vdots$ & $\vdots$ & $\vdots$\\
n & $S_{n,m_n-1}$ & $C_n$ & $Y_{n,m_n}$ & 0 & $m_n$ & $\overline{BMI}_{n,m_n}$\\
\end{tabular}
}
\end{table}
Models were fitted with four scenarios of errors (all with mean zero and variance one), and comparison was always with C\&S assuming normal errors. The four scenarios are standard normal errors,  shifted exponential with rate one, uniform$[-a,a]$ where $a= \sqrt{3}$, and log-normal distribution.

For numerical stability, the implementation of equations (\ref{equation.1.19}) and (\ref{equation.1.21}) require some minor adjustments in ``edge cases'', that is, some individual terms may return values of \texttt{NA} (``not available'', that is, a missing value) or \texttt{NaN} (``not a  number''), which can propagate to the solution and prevent convergence, e.g. if a denominator is zero or if an entry is missing. Those problematic terms are mapped to zero, which perturbs the equation but  prevents failures of convergence.


As explained in section 5, $\mu_{ij}(\theta)$ is given by equation (\ref{eqn.7.1}) and the variance function $V_{ij}(\theta)$ is defined by equation (\ref{eqn.7.2}). Tables ~\ref{parset1} to ~\ref{parset3} present the results for combinations of $n \in \{50,200\}$ and $C_{\max} \in \{125, 225\}$. The case $n=50, C_{\max}=125$ is presented in Table ~\ref{tablelabe2}, Section 5.

\begin{table}[!htbp]
\centering
\caption{Comparison with C\&S, $n=50$, $V_{ij}$ in (\ref{eqn.7.2}), $ C_{\max}=225$  }\label{parset1}
\resizebox{\columnwidth}{!}{
\begin{tabular}{lcccccccc}
  {}& ENOES=7.4{} &   {} & Our method &  {} &  {}   & {} & $C\&S$ $F_0=$ Normal & {}\\ 
  $Estimator$ & Parameter &  $|rBias|$ & ESE & ASE & {} & $|rBias|$ & ESE & ASE\\
  Normal errors & $\gamma_0$  & 0.045& 0.202 & 0.227 & {}  & 0.050 & 0.185& 0.192\\
  {$\mu_{ij}(\theta):=(\ref{eqn.7.1})$} & $\gamma_1$ &0.024& 0.263 & 0.353 & {}  & 0.054 & 0.242 & 0.270\\
  {$V_{ij}(\theta):=(\ref{eqn.7.2})$} & $\rho$ & 0.537 & 0.036 & 0.037 & {}  & 0.086& 0.028&0.033\\
  {} & $\sigma^2$  & 0.003 &0.882 & 0.900 & {} & 0.008& 0.828& 0.815\\
Exponential errors & $\gamma_0$  & 0.050 & 0.403 & 0.240 & {}  & 0.091 & 0.444& 0.199\\
  {$\mu_{ij}(\theta):=(\ref{eqn.7.1})$} & $\gamma_1$ &0.071 & 0.777 & 0.371 & {}  & 0.316 & 0.949 & 0.289\\
  {$V_{ij}(\theta):=(\ref{eqn.7.2})$} & $\rho$ & 0.321 & 0.180 & 0.036 & {}  & 0.236 & 0.068 &0.031\\
  {} & $\sigma^2$  & 0.068 & 4.718 & 1.049 & {} & 0.074 & 5.255 & 0.938\\
Uniform errors & $\gamma_0$  & 0.018 & 0.206 & 0.221 & {}  & 0.006 & 0.192 & 0.189\\
  {$\mu_{ij}(\theta):=(\ref{eqn.7.1})$} & $\gamma_1$ &0.137 & 0.308 & 0.342 & {}  & 0.004 & 0.292 & 0.273\\
  {$V_{ij}(\theta):=(\ref{eqn.7.2})$} & $\rho$ & 0.330 & 0.040 & 0.035 & {}  & 0.054 & 0.034 &0.032\\
  {} & $\sigma^2$  & 0.004 &0.558 & 0.591 & {} & 0.010& 1.365& 0.542\\
Log-normal errors & $\gamma_0$  & 0.032& 0.240 & 0.254 & {}  & 0.010 & 0.191 & 0.190\\
  {$\mu_{ij}(\theta):=(\ref{eqn.7.1})$} & $\gamma_1$ &0.184 & 0.527 & 0.441 & {}  & 0.007 & 0.276 & 0.253\\
  {$V_{ij}(\theta):=(\ref{eqn.7.2})$} & $\rho$ & 0.350 & 0.112 & 0.049 & {}  & 0.015& 0.037 &0.031\\
  {} & $\sigma^2$  & 0.013 &3.302 & 2.667 & {} & 0.043 & 2.295 & 1.931\\
\end{tabular}
}
\label{tablelabe5}
\end{table}

\begin{table}[!htbp]
\centering
\caption{Comparison with C\&S, $n=200$, $V_{ij}$ in (\ref{eqn.7.2}), $ C_{\max}=225$ }\label{parset2}
\resizebox{\columnwidth}{!}{
\begin{tabular}{lcccccccc}
  {}& ENOES=7.4{} &   {} & Our method &  {} &  {}   & {} & $C\&S$ $F_0=$ Normal & {}\\ 
  $Estimator$ & Parameter &  $|rBias|$ & ESE & Bias & {} & $|rBias|$ & ESE & Bias\\
  Normal errors & $\gamma_0$  & 0.025 & 0.110 & -0.015 & {}  & 0.002 & 0.097 & -0.001\\
  {$\mu_{ij}(\theta):=(\ref{eqn.7.1})$} & $\gamma_1$ &0.180& 0.169 & 0.072 & {}  & 0.019 & 0.137& 0.008\\
  {$V_{ij}(\theta):=(\ref{eqn.7.2})$} & $\rho$ & 0.239 & 0.028 & 0.007 & {}  & 0.010& 0.017 & -0.000\\
  {} & $\sigma^2$  & 0.011 & 0.485 & 0.125 & {} & 0.013 & 0.428 & 0.141\\
Exponential errors & $\gamma_0$  & 0.032 & 0.109 & -0.019 & {}  & 0.016 & 0.131 & 0.010\\
  {$\mu_{ij}(\theta):=(\ref{eqn.7.1})$} & $\gamma_1$ &0.203 & 0.172 & 0.081 & {} & 0.011 & 0.167 & -0.004\\
  {$V_{ij}(\theta):=(\ref{eqn.7.2})$} & $\rho$ & 0.273 & 0.034 & 0.008 & {} & 0.250 & 0.128 & -0.008\\
  {} & $\sigma^2$ & 0.015 &0.912 & 0.170 & {} & 0.006 & 1.672 & -0.064\\
Uniform errors & $\gamma_0$  & 0.024 & 0.108 & -0.014 & {}  & 0.003 & 0.094  & -0.002\\
  {$\mu_{ij}(\theta):=(\ref{eqn.7.1})$} & $\gamma_1$ &0.144 & 0.170 & 0.058 & {} & 0.012 & 0.139 & -0.005\\
  {$V_{ij}(\theta):=(\ref{eqn.7.2})$} & $\rho$ & 0.257 & 0.027 & 0.008 & {} & 0.014 & 0.017 & -0.000\\
  {} & $\sigma^2$  & 0.010 &0.308 & 0.110 & {} & 0.011 & 0.265 & 0.122\\
Log-normal errors & $\gamma_0$  & 0.046 & 0.109 & -0.028 & {} & 0.003 & 0.124 & -0.002\\
  {$\mu_{ij}(\theta):=(\ref{eqn.7.1})$} & $\gamma_1$ & 0.208 & 0.166 & 0.083 & {} & 0.000 & 0.161 & 0.000\\
  {$V_{ij}(\theta):=(\ref{eqn.7.2})$} & $\rho$ & 0.232 & 0.040 & 0.007 & {} & 0.151 & 0.084 & -0.005\\
  {} & $\sigma^2$  & 0.007 & 1.631 & 0.078 & {} & 0.030 & 1.896 & -0.326\\
\end{tabular}
}
\end{table}

\begin{table}[!htbp]
\centering
\caption{Comparison with C\&S, $n=200$, $V_{ij}$ in (\ref{eqn.7.2}), $ C_{\max}=125$}\label{parset3}
\resizebox{\columnwidth}{!}{
\begin{tabular}{lcccccccc}
  {}& ENOES=3.9{} &   {} & Our method &  {} &  {}   & {} & $C\&S$ $F_0=$ Normal & {}\\ 
  $Estimator$ & Parameter &  $|rBias|$ & ESE & Bias & {} & $|rBias|$ & ESE & Bias\\
  Normal errors & $\gamma_0$  & 0.049 & 0.438 & 0.029 & {} & 0.008 & 0.249 & 0.005\\
  {$\mu_{ij}(\theta):=(\ref{eqn.7.1})$} & $\gamma_1$ &0.095 & 0.583 & -0.038 & {} & 0.021 & 0.359 & -0.008\\
  {$V_{ij}(\theta):=(\ref{eqn.7.2})$} & $\rho$ & 0.606 & 0.070 & 0.018 & {} & 0.023& 0.033& -0.001\\
  {} & $\sigma^2$  & 0.008 &0.730 & 0.089 & {} & 0.012 & 0.604 & 0.134\\
Exponential errors & $\gamma_0$  & 0.062 & 0.451 & 0.037 & {} & 0.045 & 0.232 & -0.027\\
  {$\mu_{ij}(\theta):=(\ref{eqn.7.1})$} & $\gamma_1$ &0.112 & 0.591 & -0.045 & {} & 0.086 & 0.347 & 0.034\\
  {$V_{ij}(\theta):=(\ref{eqn.7.2})$} & $\rho$ & 1.069 & 0.092 & 0.032 & {}  & 0.123 & 0.031 & -0.004\\
  {} & $\sigma^2$  & 0.008 &1.302 & 0.084 & {} & 0.024 & 1.081 & -0.264\\
Uniform errors & $\gamma_0$  & 0.064 & 0.440 & 0.039 & {} & 0.035 & 0.233 & 0.021\\
  {$\mu_{ij}(\theta):=(\ref{eqn.7.1})$} & $\gamma_1$ & 0.130 & 0.591 & -0.052 & {} & 0.075 & 0.355 & -0.030\\
  {$V_{ij}(\theta):=(\ref{eqn.7.2})$} & $\rho$ & 0.502 & 0.067 & 0.015 & {} & 0.038 & 0.029 & -0.001\\
  {} & $\sigma^2$  & 0.007 & 0.525 & 0.079 & {} & 0.011 & 0.381 & 0.124\\
Log-normal errors & $\gamma_0$  & 0.068 & 0.465 & 0.041 & {}  & 0.085 & 0.226 & -0.051\\
  {$\mu_{ij}(\theta):=(\ref{eqn.7.1})$} & $\gamma_1$ &0.106 & 0.631 &$-$ 0.042 & {} & 0.157 & 0.345 & 0.063\\
  {$V_{ij}(\theta):=(\ref{eqn.7.2})$} & $\rho$ & 0.765 & 0.358 & 0.023 & {} & 0.033 & 0.031 & 0.001\\
  {} & $\sigma^2$ & 0.016 & 2.693 & $-$0.181 & {} & 0.064 & 1.300 & -0.703\\
\end{tabular}
}
\end{table}

\end{supplement}


\end{document}